\documentclass[submission,copyright,creativecommons]{eptcs}
\providecommand{\event}{ACT 2026} %

\usepackage{iftex}

\ifpdf
  \usepackage{underscore}         %
  \usepackage[T1]{fontenc}        %
\else
  \usepackage{breakurl}           %
\fi

\usepackage[utf8]{inputenc}
\usepackage[T1]{fontenc}

\usepackage{microtype}
\usepackage[all]{nowidow}
\usepackage[lastparline]{impnattypo}
\usepackage{graphicx}

\usepackage{enumitem}

\usepackage{mathtools, amsthm, thmtools}
\usepackage{mathpartir}

\usepackage{float}

\usepackage{multicol}

\usepackage{centernot}

\usepackage[title]{appendix}

\usepackage{hyperref}
\hypersetup{
    colorlinks=true,       %
    linkcolor=blue,        %
    filecolor=blue,        %
    urlcolor=blue,         %
    citecolor=green,        %
    pdfborder={0 0 0}      %
}
\usepackage[nameinlink,capitalize,noabbrev]{cleveref}

\usepackage{babel}

\usepackage{amsmath, amssymb}

\numberwithin{equation}{section}

\usepackage{fonttable}
\usepackage{wasysym}

\usepackage{stackrel, stackengine, old-arrows}

\usepackage{mathrsfs, amsfonts, stmaryrd}

\usepackage{quiver}

\usepackage{xcolor}

\newcommand{\R}{\mathbb{R}}
\newcommand{\N}{\mathbb{N}}
\newcommand{\subto}{\twoheadrightarrow}
\newcommand{\fibto}{\twoheadrightarrow}
\newcommand{\isoto}{\overset{\sim}\to}

\newcommand{\biglens}[2]{
	 \begin{pmatrix}{\vphantom{f_f^f}#1} \\ {\vphantom{f_f^f}#2} \end{pmatrix}
}
\newcommand{\littlelens}[2]{
	 \begin{psmallmatrix}{\vphantom{f}#1} \\ {\vphantom{f}#2} \end{psmallmatrix}
}
\newcommand{\lens}[2]{
  \relax\if@display
	 \biglens{#1}{#2}
  \else
	 \littlelens{#1}{#2}
  \fi
}

\newcommand{\lensto}{\leftrightarrows}

\newcommand{\acat}[1]{\mathbf{#1}}
\newcommand{\thecat}[1]{\mathbf{#1}}

\newcommand{\Cert}{\mathbb{C}\mathbf{ert}}

\newcommand{\Lens}{\mathbb{L}\mathbf{ens}}
\newcommand{\CertLens}{{\Cert\Lens}}

\newcommand{\Moore}{\mathbb{M}\mathbf{oore}}
\newcommand{\ODE}{{\mathbb{O}\mathbf{DE}}}
\newcommand{\CertODE}{{\Cert\ODE}}
\newcommand{\Span}{{\mathbb{S}\mathbf{pan}}}

\newcommand{\Cat}{\mathcal{C}\mathsf{at}}

\newcommand{\Fib}{\mathcal{F}\mathsf{ib}}
\newcommand{\AdTriple}{\mathcal{A}\mathsf{dTp}}
\newcommand{\Dbl}{\mathcal{D}\mathsf{bl}}
\newcommand{\Tangency}{\mathcal{T}\mathsf{an}}
\newcommand{\lModr}{\ell\mathcal{M}\mathsf{od}_r}
\newcommand{\Mod}[1]{\mathcal{M}\mathsf{od}_{s, #1}}

\newcommand{\Set}{\thecat{Set}}

\newcommand{\Man}{\thecat{OpenEuc}}
\newcommand{\BunMan}{{(\Man_\bullet \simp \Man_\bullet)_\lip}}

\newcommand{\simp}{\ltimes}
\newcommand{\lip}{{\mathrm{lip}}}
\newcommand{\lexicographic}{{\mathrm{lexico}}}

\newcommand{\id}{\mathrm{id}}
\newcommand{\de}{\mathrm{d}}

\newcommand{\xto}[1]{\xrightarrow{#1}}

\newcommand{\parallelprod}{\parallel}
\DeclareMathOperator{\Lift}{\triangleright}

\newcommand{\K}{\mathcal{K}}

\theoremstyle{plain}
\newtheorem{theorem}{Theorem}[section]
\newtheorem{lemma}[theorem]{Lemma}
\newtheorem{corollary}[theorem]{Corollary}

\theoremstyle{definition}
\newtheorem{remark}[theorem]{Remark}
\newtheorem{construction}[theorem]{Construction}
\newtheorem{definition}[theorem]{Definition}
\newtheorem{example}[theorem]{Example}
\newtheorem{convention}[theorem]{Convention}

\usepackage{tikz}

\pgfdeclarelayer{edgelayer}
\pgfdeclarelayer{nodelayer}
\pgfsetlayers{edgelayer,nodelayer,main}

\usetikzlibrary{
  	cd,
  	math,
  	decorations.markings,
		decorations.pathreplacing,
  	positioning,
  	arrows.meta,
  	shapes,
		shadows,
		shadings,
  	calc,
  	fit,
  	quotes,
  	intersections,
    circuits,
    circuits.ee.IEC
  }

	\tikzcdset{arrow style=tikz, diagrams={>=To}}

\tikzset{
     oriented WD/.style={%
        every to/.style={out=0,in=180,draw},
        label/.style={
           font=\everymath\expandafter{\the\everymath\scriptstyle},
           inner sep=0pt,
           node distance=2pt and -2pt},
        semithick,
        node distance=1 and 1,
        decoration={markings, mark=at position \stringdecpos with \stringdec},
        ar/.style={postaction={decorate}},
        execute at begin picture={\tikzset{
           x=\bbx, y=\bby,
           }}
        },
     string decoration/.store in=\stringdec,
     string decoration={\arrow{stealth};},
     string decoration pos/.store in=\stringdecpos,
     string decoration pos=.7,
	 	 dot size/.store in=\dotsize,
	   dot size=3pt,
	 	 dot/.style={
			 circle, draw, thick, inner sep=0, fill=black, minimum width=\dotsize
	   },
     bbx/.store in=\bbx,
     bbx = 1.5cm,
     bby/.store in=\bby,
     bby = 1.5ex,
     bb port sep/.store in=\bbportsep,
     bb port sep=1.5,
     bb port length/.store in=\bbportlen,
     bb port length=4pt,
     bb penetrate/.store in=\bbpenetrate,
     bb penetrate=0,
     bb min width/.store in=\bbminwidth,
     bb min width=1cm,
     bb rounded corners/.store in=\bbcorners,
     bb rounded corners=2pt,
     bb small/.style={bb port sep=1, bb port length=2.5pt, bbx=.4cm, bb min width=.4cm,
bby=.7ex},
		 bb medium/.style={bb port sep=1, bb port length=2.5pt, bbx=.4cm, bb min width=.4cm,
bby=.9ex},
     bb/.code 2 args={%
        \pgfmathsetlengthmacro{\bbheight}{\bbportsep * (max(#1,#2)+1) * \bby}
        \pgfkeysalso{draw,minimum height=\bbheight,minimum width=\bbminwidth,outer
sep=0pt,
           rounded corners=\bbcorners,thick,
           prefix after command={\pgfextra{\let\fixname\tikzlastnode}},
           append after command={\pgfextra{\draw
              \ifnum #1=0{} \else foreach \i in {1,...,#1} {
                 ($(\fixname.north west)!{\i/(#1+1)}!(\fixname.south west)$) +(-
\bbportlen,0)
  coordinate (\fixname_in\i) -- +(\bbpenetrate,0) coordinate (\fixname_in\i')}\fi
              \ifnum #2=0{} \else foreach \i in {1,...,#2} {
                 ($(\fixname.north east)!{\i/(#2+1)}!(\fixname.south east)$) +(-
\bbpenetrate,0)
  coordinate (\fixname_out\i') -- +(\bbportlen,0) coordinate (\fixname_out\i)}\fi;
           }}}
     },
     bb name/.style={append after command={\pgfextra{\node[anchor=north] at
(\fixname.north) {#1};}}}
  }

\makeatletter
  {\refstepcounter{equation}\begin{mathpar}[#1]}%
  {\eqno{\hbox{\@eqnnum}}\end{mathpar}}
\makeatother

\title{Compositionality of Lyapunov functions\\via assume-guarantee reasoning}
\author{
Matteo Capucci
\institute{University of Strathclyde, UK}
\institute{Independent Researcher, IT\thanks{Supported by Advanced Research + Invention Agency grant MSAI-PRO-01.}}
\email{matteo.capucci@gmail.com}
\and
David Jaz Myers
\institute{Topos Institute, UK\thanks{Supported by Advanced Research + Invention Agency grant MSAI-PRO-14.}}
\email{davidjaz@topos.institute}
}

\def\titlerunning{Compositionality of Lyapunov functions via assume-guarantee reasoning}
\def\authorrunning{M. Capucci \& D. J. Myers}

\begin{document}
\maketitle

\begin{abstract}
    Assume-guarantee reasoning is a technique for compositional model checking in which system specifications are checked under certain assumptions on system parameters or inputs, and provide guarantees on observations of system state.
    We present a categorical framework for assume-guarantee reasoning for safety problems by viewing systems as \emph{lenses}, following our earlier work on the compositionality of generalized Moore machines. Generalized Moore machines include ordinary Moore machines, partially observable Markov (decision) processes, and systems of parameterized ODEs (control systems); our framework gives assume-guarantee reasoning specially adapted to each of these cases. In particular, we give a novel formulation of assume-guarantee reasoning for \emph{(local) input-to-state stability} ((L)ISS) Lyapunov functions on systems of parameterized ODEs.

    Our framework is categorically natural and straightforwardly compositional. A flavor of generalized Moore machine is determined by a \emph{tangency}: a fibration with a section. We show that symmetric monoidal loose right modules of assume-guarantee certified generalized Moore machines over symmetric monoidal double categories of certified wiring diagrams can be constructed 2-functorially from fibrations internal to the 2-category of tangencies.
\end{abstract}

\section{Introduction}

Model checking aims to verify that a dynamical system satisfies a \emph{specification} --- a predicate of its behaviors. When systems are built up modularly through coupling component subsystems, the resulting composite system can be much bigger and much more difficult to check than its components. This naturally suggests a divide-and-conquer approach to checking coupled systems: if we know each component system satisfies its spec, and we know that the conjunction of component specs implies the composite spec, then we can check the composite model modularly.

However, component systems are necessarily \emph{open} --- they can be coupled with other components. It is very unlikely that a component will satisfy its spec \emph{simpliciter}; to check a component, we must take for granted some conditions on its operating environment. This observation lead Pneuli to formulate his \emph{assume-guarantee reasoning} \cite{pneuliAssumeGuarantee} method for compositional model checking. In assume-guarantee reasoning, each component provides the \emph{assumption} it makes of its environment, as well as the \emph{guarantee} it promises on its observable behavior, provided the assumption is met. Assume-guarantee reasoning has gone on to be a central tool in the model checking toolkit \cite{PasareanuDwyerHuth1999, ViswanathanViswanathan2001, CobleighGiannakopoulouPasareanu2003, CobleighGiannakopoulouPasareanuBarringer2008, GuptaMcMillanFu2008, BobaruPasareanuGiannakopoulou2008}.

In this paper, we put forward a categorical algebra for assume-guarantee reasoning of state predicates on \emph{generalized Moore machines}. Generalized Moore machines include ordinary Moore machines \cite{mooreGedanken}, systems of parameterized ODEs, partially observable Markov decision processes, and more general coalgebras for endofunctors. We will follow the \emph{double operadic theory of systems} approach \cite{libkind2025doubleoperadictheorysystems} and describe symmetric monoidal loose right modules of assume-guarantee certified machines over symmetric monoidal double categories of certified \emph{wiring diagrams} (or more general \emph{coupling laws}). Our focus in this paper is on the abstract algebra of assume-guarantee reasoning; we look forward to future work on practical implementations which exploit this compositional algebra for actual model checking.

Our approach begins with the observation that assume-guarantee reasoning is, at least for certain specifications, highly reminiscent of the definition of a Moore machine. A Moore machine with action set $A$ and observation set $O$ consists of a set $S$ of states together with two functions: the \emph{view} $v : S \to O$ yielding an output symbol and the \emph{update} $u : S \times I \to S$ transitioning the state on receiving an input. Suppose we have a predicate $\varphi(s)$ telling us which states $s \in S$ are `safe'; our aim is to prove that when starting in a safe state, we will remain in a safe state so long as our input $i \in I$ always satsifies its assumptions $\alpha(i)$. We also provide a guarantee $\gamma(o)$ on output symbols $o \in O$, which holds under the assumption that we are in a safe state. In total, we analyze assume guarantee for specifications of safety predicates of state into a pair of implications:
\begin{equation}
\lens{u}{v} \vDash \lens{\varphi}{\varphi} \leftrightarrows \lens{\alpha}{\gamma} \iff
\begin{cases}
    \varphi(s) \wedge \alpha(i) \Rightarrow \varphi(u(s,i)), \\
    \varphi(s) \Rightarrow \gamma(v(s))
\end{cases}
\end{equation}
These have almost the same form as the functions $u : S \times I \to S$ and $v : S \to O$. In fact, we can see assume-guarantee certified Moore machines (for state safety specifications) as Moore machines \emph{in another category} --- this time, a category of sets equipped with a predicate, rather than just sets in the un-certified case.

This simple observation becomes useful when linked up with the categorical approach to Moore machines using \emph{lenses}, initiated by Spivak, Vasilakopoulou, and collaborators \cite{vagner2015algebrasopendynamicalsystems, schultz2019dynamicalsystemssheaves, bakirtzisCategoricalSemantics}. In this point of view, not only is a Moore machine viewed as a lens, but so are the \emph{wiring diagrams} \cite{vagner2015algebrasopendynamicalsystems} describing how machines are to be coupled. Machines may be coupled according to the wiring diagram by \emph{lens composition}. Lens composition applied to lenses in the category of subsets suffices to give the proof rule that if a wiring diagram $\lens{w^{\sharp}}{w}$ is itself certified, then composing by it preserves the certification of machines:

\begin{equation}
  \label{rule1}
  \tag{COMP}
  \inferrule{
  	\lens{w^{\sharp}}{w} \vDash \lens{\alpha_1 \wedge \cdots \wedge \alpha_n}{\gamma_1 \wedge \cdots \gamma_n}\leftrightarrows \lens{\alpha}{\gamma}
  	\and
  	\forall i,\,\,\lens{u_i}{v_i} \vDash \lens{\varphi_i}{\varphi_i} \leftrightarrows \lens{\alpha_i}{\gamma_i}
  }{
  	\lens{w^{\sharp}}{w} \circ \left( \lens{u_1}{v_1} \parallelprod \cdots \parallelprod \lens{u_n}{v_n} \right) \vDash \lens{\alpha}{\gamma}
  }
\end{equation}

In \cite{vagner2015algebrasopendynamicalsystems}, the compositionality of machines and wiring diagrams is organized into an \emph{algebra} of Moore machines over an \emph{operad} of wiring diagrams. In the subsequent work of the second-named author \cite{Jaz_Myers_2021, jaz-2021-book, libkind2025doubleoperadictheorysystems}, \emph{generalized} Moore machines are organized into a symmetric monoidal loose right modules over the symmetric monoidal double categories of lenses; this added double categorical direction includes the \emph{homomorphisms} of machines, such as traces and coarse grainings. We may then ask for another proof rule: if machine $S'$ is certified, and $S'$ coarse-grains (or simulates) $S$ via the homomorphism $\sigma:S \to S'$, then $S$ is also certified:

\begin{equation}
  \label{rule2}
  \tag{SUBST}
  \inferrule{
    \left(\sigma, \lens{f_{\sharp}}{f}\right) : \lens{u}{v} \rightrightarrows \lens{u'}{v'}
    \and
    \lens{u'}{v'} \vDash \lens{\varphi}{\varphi} \leftrightarrows \lens{\alpha}{\gamma}
  }{
    \lens{u}{v} \vDash \lens{\varphi \circ \sigma}{\varphi \circ \sigma} \leftrightarrows \lens{\alpha \circ f_{\sharp}}{\gamma \circ f}
  }
\end{equation}

Taking lenses in the category of subsets is not quite enough to handle this rule for the most general kind of maps considered in \cite{Jaz_Myers_2021, jaz-2021-book}. However, a slight generalization which we will shortly describe does suffice.

The main contention of \cite{libkind2025doubleoperadictheorysystems} is that the compositionality of various sorts of dynamical systems is well captured by the 2-algebraic structure of a symmetric monoidal loose right module (of systems) over a symmetric monoidal double category (of interfaces and coupling laws). We'll refer to these symmetric monoidal loose right modules as \emph{2-algebras of systems} for short. In this paper, we'll show that assume-guarantee certificates may be fibered over generalized Moore machines, \emph{compositionally}; that is, we'll show that 2-algebras of certified machines are fibered over 2-algebras of machines. The right action of the double category of certified lenses will correspond to Rule \ref{rule1}, while the fibration will give us Rule \ref{rule2}. In particular, we'll find a compositional algebra for assume-guarantee reasoning about Moore machines, POMDPs, and even local input-to-state stability (LISS) Lyapunov functions on parameterized systems of ODEs.

\medskip
\textbf{Related work.}
Assume-guarantee reasoning for Moore machines has been studied \cite{gumbergModelChecking, HenzingerQadeerRajamaniTasiran2002}, though not to our knowledge using categorical methods. Our approach to compositional Lyapunov functions relates to categorical Lyapunov theory \cite{ames2025categoricallyapunovtheoryi, ames2025categoricallyapunovtheoryii}.

\medskip
\textbf{Acknowledgements.}
We would like to thank Mirco Giacobbe, Diptarko Roy, Grigory Neustroev, David Corfield, and Eigil Rischel for helpful discussions. The second-named author would like to thank James Fairbanks for suggesting the task of theorizing compositional Lyapunov functions.

\section{Certified Moore machines as lenses}\label{sec:moore.machine}

Let's expand on the observation from the introduction that assume-guarantee certified Moore machines are Moore machines in a category of subsets.

\subsection{A review of generalized Moore machines as lenses}

A Moore machine, traditionally speaking, is a pair of functions $u : S \times A \to S$ and $v : S \to O$; we think of $O$ as the set of possible \emph{observations} which may be made of it, and $A$ as a menu of \emph{actions} which may be taken. The pair $\lens{A}{O}$ is called the \emph{interface} of the Moore machine. This is a \emph{deterministic} Moore machine; in a non-deterministic Moore machine, $u$ has signature $S \times A \to \mathcal{P}S$ into the power-set of $S$. A partially observable Markov process is a Moore machine, but the update is \emph{stochastic}: we have $u : S \times A \to DS$, where $D$ is a set of probability distributions. If we want a Markov \emph{decision} process, we can also return a distribution over rewards $u : S \times A \to D(\mathbb{R} \times S)$.

We will work with a generalization of the Moore machine notion which handles these and other examples. First, we generalize the interface by allowing the menu $A_o$ of actions to depend on the observation $o \in O$; this gives us an interface $\lens{A_o}{o : O}$ consisting of the set of observations and the family $A:O \to \Set$ of action sets depending on it. We will also suppose that, for every set $S$, we have a family of sets $T_s S$ of \emph{changes} that it is possible to make in the state $s \in S$. In this setting, a \emph{generalized Moore machine} is a pair of maps:
\begin{equation}
\begin{cases}
 u : (s \in S) \times A_{v(s)} \to T_s S \\
 v : S \to O
\end{cases}
\end{equation}
where $(s \in S) \times A_{v(s)}$ is the set of pairs $(s, a)$ with $a \in A_{v(s)}$. The \emph{view} $v : S \to O$ yields an observation of state, and the update $u : (s \in S) \times A_{v(s)}\to T_sS$ applies an action which is compatible with the current observation, yielding a change $u(s, a) \in T_s S$ to state $s$. In short, we write $\lens{u}{v} : \lens{T_sS}{s : S} \leftrightarrows \lens{A_o}{o : O}$ for the Moore machine as a whole.

Taking $T_s S := S$ and $A_o := A$ to be constant, we recover the usual, `basic' notion of Moore machine; we call an interface $\lens{A}{O}$ with constant action set a \emph{simple interface}. If $T_s S := \mathcal{P}S$, we recover a non-deterministic Moore machine; if $T_s S := D S$, we recover a partially observable Markov process. Similarly, if $F$ is any endofunctor, then we may define $T_s S := FS$; generalized Moore machines are then an \emph{open} variant of $F$-coalgebras (see e.g. \cite{RUTTEN20003}). We will see in the upcoming \cref{sec:cert-odes} that taking $T_s S$ to be the tangent space of a manifold $S$ gives parameterized systems of ordinary differential equations; but we will leave that example for later.

Moore machines are made to be coupled synchronously. The simplest synchronous coupling of Moore machines is the \emph{parallel product}, in which both machines run simultaneously without interacting. Given $\lens{u_1}{v_1} : \lens{T_{s_1}S_1}{s_1 : S_1} \leftrightarrows \lens{A_{1,o_1}}{o_1 : O_1}$ and $\lens{u_2}{v_2} : \lens{T_{s_2}S_2}{s_2 : S_2} \leftrightarrows \lens{A_{2,o_2}}{o_2 : O_2}$, their parallel product $\lens{u_1}{v_1} \parallelprod \lens{u_2}{v_2} : \lens{T_{(s_1, s_2)}(S_1 \times S_2)}{(s_1, s_2) : S_1 \times S_2} \leftrightarrows \lens{A_{1,o_1} \times A_{2,o_2}}{(o_1, o_2) : O_1 \times O_2}$ is defined by:
\begin{equation}
	\lens{u_1}{v_1} \parallelprod \lens{u_2}{v_2} := \begin{cases}
		((s_1, s_2), (a_1, a_2)) \mapsto \mu_T(u_1(s_1, a_1), u_2(s_2, a_2)) \\
		(s_1, s_2) \mapsto (v_1(s_1), v_2(s_2))
	\end{cases}
\end{equation}
where $\mu_T : T_{s_1}S_1 \times T_{s_2}S_2 \to T_{(s_1, s_2)} (S_1 \times S_2)$ is a way of turning a pair of changes into a change of a pair. For the basic case $T_s S := S$, $\mu_T(s_1, s_2) = (s_1, s_2)$ is the identity: a pair of changes is already a change of a pair. For $T_s S := FS$ the case of a coalgebra, $\mu_T : FS_1 \times FS_2 \to F(S_1 \times S_2)$ requires $F$ to be a \emph{lax monoidal functor}. For $F = \mathcal{P}$ we may define $\mu_{\mathcal{P}}(U_1, U_2) := U_1 \times U_2$ to send a pair of subsets to a subset of pairs; for $F = D$ we may define $\mu_{D}(p_1, p_2)$ to be the independent product of distributions.

Far more interesting than the parallel product is an actual coupling of Moore machines where the component machines interact. When the interfaces invovled are simple, such a coupling can be described by a \emph{(directed) wiring diagram} \cite{rupel2013operadtemporalwiringdiagrams}. Here are three basic wiring diagrams which will form our running examples:

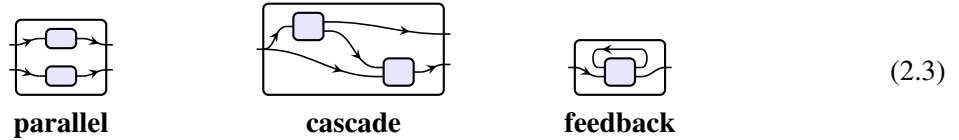
\begin{equation}
		\begin{array}{ccc}
\begin{tikzpicture}[oriented WD, every fit/.style={inner xsep=\bbx, inner ysep=\bby}, bb small]
  \node[bb={1}{1}, fill=blue!10] (X1) {};
  \node[bb={1}{1}, fill=blue!10, below = 2 of X1] (X2) {};
  \node[bb={2}{2}, fit={(X1) (X2)}] (Z) {};
\draw[ar] (Z_in1') to (X1_in1);
\draw[ar] (Z_in2') to (X2_in1);
\draw[ar] (X1_out1) to (Z_out1');
\draw[ar] (X2_out1) to (Z_out2');
  \end{tikzpicture}
			\qquad & \qquad
\begin{tikzpicture}[oriented WD, every fit/.style={inner xsep=\bbx, inner ysep=\bby}, bb small]
  \node[bb={1}{2}, fill=blue!10] (X1) {};
  \node[bb={2}{1}, fill=blue!10, below right = 2 and 2 of X1] (X2) {};
  \node[bb={1}{2}, fit={(X1) (X2)}] (Z) {};
\draw[ar] (Z_in1') to (X1_in1);
\draw[ar] (Z_in1') to (X2_in2);
\draw[ar] (X1_out2) to (X2_in1);
\draw[ar] (X1_out1) to (Z_out1');
\draw[ar] (X2_out1) to (Z_out2');
  \end{tikzpicture}

			\qquad & \quad

\begin{tikzpicture}[oriented WD, every fit/.style={inner xsep=\bbx, inner ysep=\bby}, bb small]
  \node[bb={2}{2}, fill=blue!10] (X1) {};
  \node[bb={1}{1}, fit={(X1) ($(X1.north)+(0,1)$)}] (Y1) {};
  \draw[ar] (Y1_in1') to (X1_in2);
  \draw (X1_out2) to (Y1_out1');
  \draw[ar] let \p1=(X1.north east), \p2=(X1.north west), \n1={\y1+\bby}, \n2=\bbportlen in
          (X1_out1) to[in=0] (\x1+\n2,\n1) -- (\x2-\n2,\n1) to[out=180] (X1_in1);
\end{tikzpicture}
			\\
			\begin{gathered}
				\textbf{parallel}\\
			\end{gathered}
			\qquad & \qquad
			\begin{gathered}
				\textbf{cascade}\\
			\end{gathered}
			\qquad & \quad
			\begin{gathered}
				\textbf{feedback}\\
			\end{gathered}
		\end{array}
	\end{equation}

In order to describe how wiring diagrams act on Moore machines, we must first make three observations:

\smallskip
\noindent\textbf{1.} Every wiring diagram determines a \emph{lens}.\footnote{In fact, wiring diagrams may be \emph{defined} as lenses in free cartesian categories, see Chapter 1.4 of \cite{jaz-2021-book}.} A lens $\lens{w^{\sharp}}{w} : \lens{A_{1,o_1}}{o_1 : O_1} \leftrightarrows \lens{A_{2,o_2}}{o_2: O_2}$ is a pair of maps
\begin{equation}
\begin{cases}
w^{\sharp} : (o_1 \in O_1) \times A_{w(o_1)} \to A_{1,o_1} \\
w : O_1 \to O_2
\end{cases}.
\end{equation}
The above diagrams correspond to the following simple lenses:

\hspace{-1cm}
\begin{minipage}{\textwidth}
  \begin{equation*}
    \begin{array}{ccc}
      \begin{gathered}
        {\textstyle 		\textbf{parallel}:
          \lens{A_1}{O_1} \parallelprod \lens{A_2}{O_2} \leftrightarrows \lens{A_1 \times A_2}{O_1 \times O_2}
        }			\end{gathered}
      & \begin{gathered}
          {\textstyle				\textbf{cascade}:
            \lens{A}{O_1 \times M} \parallelprod \lens{M \times A}{O_2} \leftrightarrows \lens{A}{O_1 \times O_2}
          }			\end{gathered}
      & \begin{gathered}
          {\textstyle				\textbf{feedback}:
           \lens{A \times M}{M \times O} \to \lens{A}{O}
          }			\end{gathered}
      \\
      \begin{cases}
        ((o_1, o_2), (a_1, a_2)) \mapsto (a_1, a_2) \\
        (o_1, o_2) \mapsto (o_1, o_2)
      \end{cases} &
      \begin{cases}
        (((o_1, m), o_2), a) \mapsto (a, (m, a)) \\
        ((o_1, m), o_2) \mapsto (o_1, o_2)
      \end{cases}    &
      \begin{cases}
        ((m, o), a) \mapsto (a, m) \\
        (m, o) \mapsto o
      \end{cases}
    \end{array}
  \end{equation*}
\end{minipage}

\smallskip
\noindent\textbf{2.} Moore machines are themselves lenses of the special form $\lens{u}{v} : \lens{T_sS}{s : S} \leftrightarrows \lens{A_o}{o : O}$.
\smallskip

\noindent\textbf{3.} We may \emph{compose} lenses $\lens{w^{\sharp}}{w} :
\lens{A_{1,o_2}}{o_1 : O_1} \leftrightarrows \lens{A_{2,o_2}}{o_2 : O_2}$ and $\lens{t^{\sharp}}{t} :
\lens{A_{2, o_2}}{o_2 : O_2} \leftrightarrows \lens{A_{3, o_3}}{o_3 : O_3}$ by the formulas
\begin{equation}
	\lens{t^{\sharp}}{t} \circ \lens{w^{\sharp}}{w} := \begin{cases}
		(o_1, a_3) \mapsto w^{\sharp}(o_1, t^{\sharp}(w(o_1), a_3)) \\
		o_1 \mapsto t(w(o_1))
	\end{cases}
\end{equation}

With these observations in mind, we may couple systems $\lens{u_i}{v_i} : \lens{T_{s_i}S_i}{s_i : S_i} \leftrightarrows \lens{A_{i,o_i}}{o_i : O_i}$ according to a wiring diagram (or more generally a \emph{coupling law}) $\lens{w^{\sharp}}{w} : \lens{A_{1,o_1}}{o_1:O_1} \parallelprod \cdots \parallelprod \lens{A_{n,o_n}}{o_n : O_n} \leftrightarrows \lens{A_o}{O : O}$ by taking the composite $\lens{w^{\sharp}}{w} \circ \left(\lens{u_1}{v_1} \parallelprod \cdots \parallelprod \lens{u_n}{v_n}\right)$.
The main takeaway here is that we have \emph{objectified} the coupling law into a lens; this objective coupling law then \emph{acts} upon systems by composition on the right, coupling them together. For example, the \textbf{cascade} wiring diagram acts by taking the cascade product of Moore machines (see e.g. Definition 1 of \cite{geatti2025automata}).

Maps of Moore machines are also important. See Chapter 3 of \cite{jaz-2021-book} for examples.

\subsection{Certified Moore machines as certified lenses}

As we noticed in the introduction, the setup of assume-guarantee reasoning --- with its specificiation $\varphi$ of states, assumption $\alpha$ on actions, and guarantee $\gamma$ on observations --- has a form highly reminiscent of a Moore machine itself. We will see that there is a strong wiff of the \emph{microcosm principle} at work in assume-guarantee reasoning.

We begin by associating the assumptions and guarantees to an interface. We may take the guarantee $\gamma : O \to \mathsf{bool}$ to be a predicate on obeservations, and we need an assumption $\alpha_o : A_o \to \mathsf{bool}$ on each set of actions; furthermore, we require that $\alpha_o(a) \Rightarrow \gamma_o$ for all $o \in O$.\footnote{This implication condition may seem somewhat mysterious, but it is clarified by viewing propositions as types: the morally correct signature of the assumption is $\alpha : (o : O) \times A_o \times \gamma(o) \to \mathsf{Prop}$. Just as the set of actions $A_o$ depends on the observation, the type of certificates that the assumption holds depends on a certificate that the guarantee holds of a given observation.} We write $\lens{A_o}{o : O} \vDash \lens{\alpha}{\gamma}$ to mean that the interface is declared to satisfy the stated assumption and guarantee.

In a simple interface where $A_o = A$, we may take $\alpha(o, a) := \gamma(o) \wedge \overline{\alpha}(a)$ for a fixed assumption $\overline{\alpha} : A \to \mathsf{bool}$; the more general form of $\alpha(-,-)$ in the dependent case is to ensure stability when substituting by a map of interfaces. If $\lens{f_{\sharp}}{f} : \lens{A_{1, o_1}}{o_1 : O_1} \rightrightarrows \lens{A_{2, o_2}}{o_2 : O_2}$ is a map of interfaces, and $\lens{A_{2, o_2}}{o_2 : O_2} \vDash \lens{\alpha}{\gamma}$, then we may pull back the assumption and guarantee,
\begin{equation}
\lens{A_{1, o_1}}{o_1 : O_1} \vDash \lens{f_{\sharp}}{f}^{\ast}\lens{\alpha}{\gamma} := \begin{cases}
 (o_1, a_1) \mapsto \alpha(f(o_1), f_{\sharp}(o_1, a_1)) \\
  o_1 \mapsto \gamma(f(o_1))
\end{cases}.
\end{equation}
In this way, assumptions and guarantees are \emph{fibred} over interfaces. Note that even if all interfaces involved are simple, subsitution into $\alpha(o_2, a_2) := \gamma(o_2) \wedge \overline{\alpha}(a_2)$ yields $\gamma(f(o_1)) \wedge \overline{\alpha}(f_{\sharp}(o_1, a_1))$ which is no longer of the form $\gamma(o) \wedge \alpha(a)$; this is why we must allow for the more general form of assumption. \footnote{For example, substituting by a chart $\lens{a}{o} : \lens{\{\mathsf{tick}\}}{\mathbb{N}} \rightrightarrows \lens{A_o}{o : O}$ gives the predicates $\gamma(o_n)$ and $\gamma(o_n) \wedge \overline{\alpha}(a_n)$ of $n \in \mathbb{N}$ that the guarantee and assumption hold of the sequence of observations and actions $\lens{a}{o}$, even though there is only one possible action $\mathsf{tick}$ for the clock $\lens{+1}{\id} : \lens{\N}{\N} \leftrightarrows \lens{\{\mathsf{tick}\}}{\N}$.}

We aim to show that a state $s \in S$ satisfies a \emph{spec} $\varphi : S \to \mathsf{bool}$, and the assumption hold of an action $a$, then the next state $u(s, a)$ will also satisfy the spec. In general, $u(s, a) \in T_s S$ is a \emph{change} of state; therefore, we need a way to \emph{lift} a predicate $\varphi$ of states to a predicate $\Lift_s \varphi : T_s S \to \mathsf{bool}$ of changes to state. For the classical case $T_s S := S$, we may take $\Lift_s \varphi := \varphi$; for a general endofunctor $T_s S := FS$, we must use a \emph{predicate lifting} associated to $F$ (see Definition 6.1.1 of \cite{Jacobs_2016}, or \cite{ulrichPredicateLifting}). We may then say that a system $\lens{u}{v} : \lens{T_s S}{s : S} \leftrightarrows \lens{A_o}{o : O}$ is \textbf{certified} as follows:
\begin{equation}
\lens{u}{v} \vDash \lens{\Lift \varphi}{\varphi} \leftrightarrows \lens{\alpha}{\gamma} \iff \begin{cases}
	\varphi(s) \wedge \alpha(v(s),a) \Rightarrow \Lift_{s} \varphi(u(s, a)) \\
	\varphi(s) \Rightarrow \gamma(v(s))
\end{cases}
\end{equation}

We may observe that this is a special case of the certification of lenses. If $\lens{w^{\sharp}}{w} : \lens{A_{1, o_1}}{o_1 : O_1} \leftrightarrows \lens{A_{1, o_1}}{o_1 : O_1}$ and $\lens{A_{i, o_i}}{o_i : O_i} \vDash \lens{\alpha_i}{\gamma_i}$, then
\begin{equation}
\lens{w^{\sharp}}{w} \vDash \lens{\alpha_1}{\gamma_1} \leftrightarrows \lens{\alpha_2}{\gamma_2} \iff \begin{cases}
	\gamma_1(o_1) \wedge \alpha_2(w(o_1), a_2) \Rightarrow \alpha_1(o_1, w^{\sharp}(o_1, a_2)) \\
	\gamma_1(o_1) \Rightarrow \gamma_2(w(o_1))
\end{cases}.
\end{equation}

Intuitively, thinking of $\lens{w^{\sharp}}{w}$ as a wiring diagram, the guarantees of all inner boxes must imply the guarantee of the outer box, and the inner guarantees and outer assumption must together imply the inner assumptions.
For the \textbf{cascase} wiring diagram $\lens{w^{\sharp}}{w} : \lens{A}{O_1 \times M} \parallelprod \lens{M \times A}{O_2} \leftrightarrows \lens{A}{O_1 \times O_2}$, proving that $\lens{w^{\sharp}}{w} \vDash \lens{\alpha_1}{\gamma_1} \parallelprod \lens{\alpha_2}{\gamma_2} \leftrightarrows \lens{\alpha_3}{\gamma_3}$
explicitly means validating that
\begin{equation}
\begin{cases}
	\gamma_1(o_1, m) \wedge \gamma_2(o_2) \wedge \alpha_3((o_1, o_2), a) \Rightarrow \alpha_1((o_1, m), a) \wedge \alpha_2(o_2, (m, a))\\
	\gamma_1(o_1, m) \wedge \gamma_2(o_2) \Rightarrow \gamma_3(o_1, o_2)
\end{cases}
\end{equation}
Supposing that $\alpha_i(o, a) = \gamma_i(o) \wedge \overline{\alpha}_i(a)$ is of the simple form, it then suffices to show that $\overline{\alpha}_3(a) \Rightarrow \overline{\alpha}_1(a)$, that $\gamma_1(o_1, m) \wedge \overline{\alpha}_3(a) \Rightarrow \overline{\alpha}_2(m, a)$, and that $\gamma_1(o_1, m) \wedge \gamma_2(o_2) \Rightarrow \gamma_3(o_1, o_2)$. These implications are suggested by the form of the wiring diagram (which always produces a simple lens).

Checking that certification of lenses is closed under composition of lenses and the parallel product gives us the proof rule \eqref{rule1} from the introduction; checking that certification of systems is closed under substitution along maps of systems gives us proof rule \eqref{rule2}.
In total, knowing that the 2-algebra of certified machines is fibered over that of generalized Moore machines gives us a compositional algebra for assume-guarantee reasoning.

\section{Certified Moore machines from fibrations of tangencies}\label{sec:gen.moore.machine}

In this section, we formalize the above story, giving a compact description of all of the data involved and a method to check the properties required for the proof rules. We follow \cite{libkind2025doubleoperadictheorysystems} in defining generalized Moore machines internal to a \emph{tangency}.\footnote{What we here call a tangency was called a ``doctrine of dynamical systems'' in \cite{Jaz_Myers_2021}.}

\begin{definition}[Definition 7.1 of \cite{libkind2025doubleoperadictheorysystems}]
A \emph{tangency} consists of:
\begin{enumerate}
\item A category $\acat{B}$ whose objects we think of as spaces of states or observations,
\item A (cloven) Grothendieck fibration $\pi : \acat{E} \to \acat{B}$ whose fibers $\acat{E}_O$ over a space $O \in \acat{B}$ we think of as `bundles $A_o$ of possible actions given an observation $o$'. We will still denote objects $A \in \acat{E}$ as pairs $\lens{A_o}{o : O}$ where $O = \pi A$ and maps in $\acat{E}$ as pairs $\lens{f_{\sharp}}{f} : \lens{A_o}{o : O} \rightrightarrows \lens{A'_{o'}}{o' : O'}$ where $f$ is the value under $\pi$ and $f_{\sharp}$ is the vertical factor. We think of $f_{\sharp}$ as having signature $(o : O) \times A_o \to A_{f(o)}$.
\item A section $T : \acat{B} \to \acat{E}$ assigning to each `state space' $S$ its bundle $\lens{T_sS}{s : S}$ of \emph{changes}.
\end{enumerate}
The 2-category of tangencies consists of a cartesian functor between fibrations and a colaxitor on sections. A \emph{symmetric monoidal tangency} is a pseudo-monoid in this 2-category; equivalently, it is a tangency where $\pi : \acat{E} \to \acat{B}$ is a strict monoidal fibration and $T : \acat{B} \to \acat{E}$ is lax monoidal. We will always write the monoidal product as $\parallelprod$, and will generally assume it is cartesian in both $\acat{E}$ and $\acat{B}$.
\end{definition}

\begin{definition}[Definition 2.35 of \cite{libkind2025doubleoperadictheorysystems}]
 The double category $\Lens(\pi) := \mathbb{S}\mathbf{pan}(\acat{E}, (\mathsf{vert}, \mathsf{cart}))$ of Spivak lenses \cite{spivak2022generalizedlenscategoriesfunctors} in a fibration $\pi : \acat{E} \to \acat{B}$ is the double category of spans in $\acat{E}$ whose left leg is vertical and whose right leg is cartesian, composing by pullback. (See Definition 3.8 and Theorem 3.9 of \cite{myers2021cartesianfactorizationsystemsgrothendieck} for a direct comparison with \cite{spivak2022generalizedlenscategoriesfunctors}).
\end{definition}

\begin{definition}\label{defn:moore.machine}
 Let $(\pi, T)$ be a tangency. A $(\pi, T)$-Moore machine is a $pi$-lens $$\lens{T_s S}{s : S} \xleftarrow{\lens{u}{\id}} \lens{A_{v(s)}}{s : S} \xrightarrow{\lens{\id}{v}} \lens{A_o}{o : O}.$$
A morphism of $(\pi, T)$-machines is a square in $\Lens(\pi)$ whose left leg is of the form $T\sigma$.
\end{definition}

The category $\Moore(\pi, T)$ forms a loose right module of the double category $\Lens(\pi)$ by composition on the right; if $(\pi, T)$ is symmetric monoidal then $\Moore(\pi, T)$ is a symmetric monoidal loose right module. The 2-algebraic structure of the symmetric monoidal loose right module $\Moore(\pi, T)$ encodes the compositionality of generalized Moore machines, as well as of their behaviors (via representable functors, as in Theorem 6.2 of \cite{Jaz_Myers_2021} or Section 5.3 of \cite{jaz-2021-book}).

Moreover, the assignment $(\pi, T) \mapsto \Moore(\pi, T)$ is 2-functorial from the 2-category of tangencies to the 2-category of loose right modules (Section 7.2 of \cite{libkind2025doubleoperadictheorysystems}). In \cref{app:fibs-of-tans}, we show that this 2-functor futhermore preserves \emph{fibrations of tangencies}.

\begin{theorem}\label{thm:moore.preserves.fibrations}
The 2-functors $\Lens : \Fib \to \Dbl$ and $\Moore : \Tangency \to \lModr$ preserve fibrations (as sketched in Example 5.15 of \cite{arkor2024enhanced2categoricalstructurestwodimensional}); in particular, a fibration of tangencies induces a fibration of loose right modules.
\end{theorem}

In order to give a compositional algebra for assume-guarantee certified $(\pi, T)$-machines, it therefore suffices to give a fibration over the tangency $(\pi, T)$.

\begin{lemma}\label{lem:fibration.of.tangencies}
Let $\pi : \acat{E} \twoheadrightarrow \acat{B}$ and $T : \acat{B} \to \acat{E}$ form a tangency. A fibration of tangencies over $(\pi, T)$ is given by:
\begin{enumerate}
\item Another tangency $\acat{P}\pi : \acat{PE} \twoheadrightarrow \acat{PB}$ with section $\Lift : \acat{PB} \to \acat{PE}$.
	\item A \emph{strict} morphism of tangencies:
	\begin{equation}
\begin{tikzcd}[row sep=small]
	{\acat{PB}} & {\acat{PE}} & {\acat{PB}} \\[1.5ex]
	{\acat{B}} & {\acat{E}} & {\acat{B}}
	\arrow["\Lift", from=1-1, to=1-2]
	\arrow["{p_{\acat{B}}}"', two heads, from=1-1, to=2-1]
	\arrow["{\acat{P}\pi}", two heads, from=1-2, to=1-3]
	\arrow["{p_{\acat{E}}}"', two heads, from=1-2, to=2-2]
	\arrow["{p_{\acat{B}}}"', two heads, from=1-3, to=2-3]
	\arrow["T"', from=2-1, to=2-2]
	\arrow["\pi"', two heads, from=2-2, to=2-3]
\end{tikzcd}
	\end{equation}
	Such that $p_{\acat{B}}$ and $p_{\acat{E}}$ are both fibrations, and $p_{\acat{E}}$ sends chosen $\acat{P}\pi$-cartesian morphisms to chosen $\pi$-cartesian morphisms (strictly preserving the cleavage). Moreover, $\Lift$ must be cartesian as well, though it does not need to strictly preserve the cleavage.
		\end{enumerate}
		A fibration of monoidal tangencies in addition requires that $p_{\acat{B}}$ and $p_{\acat{E}}$ be (strict) monoidal fibrations and furthermore that the above diagram commutes in the 2-category of symmetric monoidal categories and lax symmetric monoidal functors.
\end{lemma}

We think of $\acat{PB}_O$ as the category of predicates concerning an observation $o : O$, and $\acat{PE}_{\lens{A_o}{o : O}}$ as predicates concerning actions $a : A_o$. The fibration $\acat{P}\pi$ witnesses that predicates on actions should depend on predicates on observations; this dependency is important to keep track of when considering explicit certificates, and not just the truth of predicates.

Finally, we give a construction which directly captures the story developed in \cref{sec:moore.machine} for ordinary Moore machines.

\begin{lemma}\label{lem:example.fibration.tangency}
Let $\thecat{PSet} := \int^{X \in \thecat{set}} \thecat{Set}(X, \mathsf{bool})$ denote the category of sets equipped with a predicate. Then
\begin{equation}
\begin{tikzcd}[sep = small]
	{\thecat{PSet}} & {\thecat{PSet}^{\downarrow}} & {\thecat{PSet}} \\[1.5ex]
	{\thecat{Set}} & {\thecat{Set}^{\downarrow}} & {\thecat{Set}}
	\arrow["\Lift", from=1-1, to=1-2]
	\arrow["p"', two heads, from=1-1, to=2-1]
	\arrow["{\mathsf{cod}}", two heads, from=1-2, to=1-3]
	\arrow["{p^{\downarrow}}"', two heads, from=1-2, to=2-2]
	\arrow["p", two heads, from=1-3, to=2-3]
	\arrow["T"', from=2-1, to=2-2]
	\arrow["{\mathsf{cod}}"', two heads, from=2-2, to=2-3]
\end{tikzcd}
\end{equation}
is a cartesian monoidal fibration of tangencies, where
\begin{equation}
  TS = (\pi_1 : S \times S \to S),
  \qquad\quad
  \Lift(S, \varphi) = (\pi_1 : (S \times S, \varphi \times \varphi) \to (S, \varphi)).
\end{equation}
\end{lemma}

\section{Certifying the stable equilibria of open ODEs}
We now exhibit a \emph{quantitative} instance of certified generalized Moore machines. Though it escapes being captured by the general method of \cref{sec:gen.moore.machine}, it still sees the algebra of  compositional verification expressed by a fibration of algebras of systems.
In this section we will capture the algebraic structure of the \emph{Lyapunov method} in the stability theory of ODEs.

\subsection{A review of LISS}
\label{sec:liss}

\begin{definition}[{Open manifolds}]
  Let $\Man$ be the category whose objects are open subsets of cartesian spaces (which are those of the form $\R^n$ for $n \in \N$, including $\R^0 = \{\ast\}$) with a $C^1$ manifold structure, and whose maps are $C^{1,1}$-functions, i.e. once-differentiable functions with Lipschitz derivative. We will call these \emph{open manifolds}.
  We also consider the category of \textbf{pointed open manifolds} $\Man_\bullet$ and \emph{definite} maps: $f(x) = y_0$ if and only if $x = x_0$.
\end{definition}

\begin{convention}
  From now on, we tacitly point all our manifolds and functions, adopting the notational convention that a pointed open manifold $X$ is pointed by $x_0$, but also assuming, without loss of generality, that $x_0 = 0$.
\end{convention}

Consider an open ODE $\dot x = f(x,a)$ on an open manifold $X$, where $A$ is an open manifold of \emph{controls}, \emph{inputs}, or \emph{parameters}, such that $0 \in X$ is an equilibrium point.
\emph{(Local) Input-to-State Stability} ((L)ISS) is a property of such an equilibrium which, roughly speaking, prescribes that trajectories starting nearby should not stray far, and in fact eventually converge back to the equilibrium.
We give precise definitions slightly adapted to our setting and then show how to recover the classical ones (for which we refer to \cite{sontagInputStateStability2008,mironchenkoInputtoStateStabilityTheory2023}).

We start with some technical vocabulary.
The theory of ISS makes a great use of so-called \emph{comparison functions}.

\begin{definition}[{Comparison functions, \cite[§2.4]{sontagInputStateStability2008}}]
  A \textbf{local\footnotemark~storage function} is a continuous definite map $\varphi : X \to \R$, i.e. $\varphi(x) = 0$ iff $x=x_0$.
  \footnotetext{The term ``storage function'' is from the stability theory literature. The adjective ``local'' is used here to mean we dropped the condition of \emph{radial unboundedness}.}
  A \textbf{$\K$ function} is a strictly increasing local storage function $\R^+ \to \R$.
  A \textbf{$\K_\infty$ function} is an unbounded $\K$ function. We also define $\K_\infty^0 := \K_{\infty} \cup \{0\}$ to include the constant function at $0$; we note that $\K_{\infty}^0$ is a monoid under addition (whereas $\K_{\infty}$ is only a semigroup).
\end{definition}

It is important to note that every local storage function can be approximated above and below by \emph{radially symmetric} (meaning they factor through the distance function $|x_0 - (-)| = |-|$) $\K$ functions:

\begin{lemma}[{\cite[Corollary~A.23]{mironchenkoInputtoStateStabilityTheory2023}}]
  \label{lemma:k-approx}
  If $\varphi$ is local storage, there are $\varphi^+, \varphi^- \in \K$ which are radially symmetric and such that $\varphi^- \leq \varphi \leq \varphi^+$, namely
  \begin{equation}
    \varphi^+(x) = \sup_{|x'| \leq |x|} \varphi(x'),
    \qquad
    \varphi^-(x) = \inf_{|x'| \geq |x|} \varphi(x').
  \end{equation}
  Moreover, if $\varphi$ is unbounded, so are $\varphi^+$ and $\varphi^-$ (i.e. $\varphi^+, \varphi^- \in \K_\infty$).
\end{lemma}

\begin{remark}\label[remark]{rmk:k-funs}
  We interpret $\K$ and $\K_\infty$ functions as describing radially symmetric homeomorphisms $f:B_r(x_0) \isoto B_r'(y_0)$ between balls of any dimension.
  That is, every such homeomorphism is determined by what it does on radii, and such a `stretching schedule' is easily seen to be a $\K$ function when $r < \infty$ and $\K_\infty$ functions when $r=\infty$.
  \emph{Vice versa}, every $\K$ function $\kappa$ induces radially symmetric homeomorphisms $B_r(x_0) \isoto B_{\kappa(r)}(y_0)$.
\end{remark}

In light of this, it is evident why the following definition is motivated as a definition of exponential stability \emph{invariant under non-linear change of variables} (cf. \cite[§2.8]{sontagInputStateStability2008}):

\begin{definition}\label[definition]{def:iss}
  We say $(X,x_0,u)$ is \textbf{ISS} when there exist $\kappa_1, \kappa_2, \kappa_3 \in \K_\infty$ such that, for every $a:\R^+ \to A$ and trajectory $x:\R^+ \to X$,
  \begin{equation}\label{eq:exp-stab}
    \forall t \geq 0, \quad
    |x(t, a)| \leq \kappa_1(\kappa_2(|x(0)|)e^{-t}) + \kappa_3(\|a\|_\infty).
  \end{equation}
  where $\|a\|_\infty = \sup_{t \geq 0} |a(t)|$.
  This definition reduces to \textbf{local ISS} (\textbf{LISS}) when both $X$ and $A$ are bounded sets, and to \textbf{global ISS} (or just \textbf{ISS}) otherwise, though note that standard global ISS is stated for $A$ and $X$ unbounded.
\end{definition}

The notion of ISS stability is originally due to Sontag \cite{sontagSmoothStabilizationImplies1989}.
It is one of the most important notions of stability used in control theory, since it is both practically relevant and theoretically convenient.
Two aspects of ISS are often praised: the fact it applies to non-linear systems as much as linear ones and that it admits a so-called \emph{dissipative} characterization.
This means one can certify ISS by exhibiting a suitable \textbf{Lyapunov function} directly on the system of ODEs, rather than solving the equation first and verifying \cref{def:iss} on the trajectories.
This is a powerful property, as it is much easier to certify behaviour \emph{at the system level} rather than first find its solutions (which might not have a closed form) and certify those.

\begin{theorem}[LISS Lyapunov]
\label{th:lyapunov-iss}
  The equilibrium of open ODE $(X,0, u)$ is \emph{local} ISS precisely when there exists a differentiable local storage function $\varphi : X \to \R$, called the \textbf{LISS Lyapunov function}, as well as $\kappa_1, \kappa_2 \in \K$ such that
  \begin{equation}
    \text{for all $a\in A$, $x \in X$,\quad $\kappa_1(|a|) \geq \de \varphi(f(x,a)) + \kappa_2(\varphi(x))$}.
  \end{equation}
\end{theorem}
\begin{proof}
  This is \cite[Theorem~3.4]{sontagInputStateStability2008}.
  There is also noted that $\varphi$ can be chosen \emph{smooth}.
\end{proof}

\begin{corollary}[ISS Lyapunov]
  The equilibrium of open ODE $(X, 0, u)$, where $X$ is unbounded, is \emph{global} ISS if it is local ISS and $\varphi$ is \emph{radially unbounded}, that is, $\varphi(x) \to \infty$ as $|x| \to \infty$.
  Such a $\varphi$ is called \textbf{ISS Lyapunov}.
\end{corollary}

\begin{corollary}\label{cor:main}
  For the equilibrium of open ODE $(X, 0, u)$ to be LISS it suffices to exhibit, besides the differentiable local storage function $\varphi : X \to \R$, a local storage function $\alpha: A \to \R$ and $\kappa \in \K_\infty$ such that
  \begin{equation}
    \text{for all $a\in A$, $x \in X$,\quad $\alpha(a) \geq \de \varphi(f(x,a)) + \kappa(\varphi(x))$}.
  \end{equation}
\end{corollary}
\begin{proof}
  In light of \cref{lemma:k-approx}, $\alpha^+$ is of the form $\lambda(|-|)$ for $\lambda \in \K$ and $\alpha^+ \geq \alpha$.
\end{proof}

We thus start to glimpse the form of the certified lenses from the previous sections.

\subsection{Certified equilibria of open ODEs}
\label{sec:cert-odes}
To capture LISS using the framework of certified Moore machines, we start by observing that `ODEs with a distinguished equilibrium point' are generalized Moore machines in the sense of \cref{defn:moore.machine}.
We define the appropriate tangency below; in this case, the section $T$ will actually be the tangent bundle.

\begin{construction}[Equilibria of open ODEs as generalized Moore machines]
  Consider again the category of pointed open manifolds $\Man_\bullet$.
  We equip it with a notion of ``bundles of actions'' given by pointed trivial topological bundles $\pi_1 :B \times F \to B$ with $F \in \Man$, and pointed fiberwise Lipschitz (``$C^{0,1}$'') maps between them.
  We thus get a (cartesian monoidal) fibration of bundles as below right:
  \begin{equation}\label{eq:man-bundles-fib}
    \begin{tikzcd}[ampersand replacement=\&]
      1 \& {B \times F} \&\& {B' \times F'} \\[-1ex]
      \& B \&\& {B'}
      \arrow["{(b_0, f_0)}", from=1-1, to=1-2]
      \arrow["{{b_0}}"', from=1-1, to=2-2]
      \arrow["{\text{fb.wise $C^{0,1}$}}", from=1-2, to=1-4]
      \arrow["\,\pi_1", two heads, from=1-2, to=2-2]
      \arrow[two heads, from=1-4, to=2-4]
      \arrow["{C^{1,1}}", from=2-2, to=2-4]
    \end{tikzcd}
    \qquad
    \begin{tikzcd}[ampersand replacement=\&]
      {\BunMan} \\[-1ex]
      {\Man_\bullet}
      \arrow["{{\,\partial_0}}", two heads, from=1-1, to=2-1]
    \end{tikzcd}
  \end{equation}
  We pick a tangency $T$ by defining it to be the usual tangent bundle assignment, which sends a manifold of dimension $n$ to $X \times \R^n \xto{\pi_X} X$ and a $C^{1,1}$ function to its fiberwise Lipschitz Jacobian. The total space $T(X,x_0)$ is pointed by the pair $(x_0, 0)$.
  We define the section $T$ to be the usual tangent bundle assignment, which sends an open manifold of dimension $n$ to $X \times \R^n$ and a $C^{1,1}$ function to its fiberwise Lipschitz Jacobian, The total space $T(X,x_0)$ is pointed by the pair $(x_0, 0)$.  We keep denoting this as $TX \subto X$.

  The generalized Moore machines associated to this tangency are \textbf{equilibria of open ODEs}.
  These are given by lenses $\lens{u}{v} : \lens{T_x X}{x:X} \lensto \lens{A}{O}$ which correspond to the differential equation
  $\dot x = u(x, a)$
  together with an observable
  $o = v(x)$
  and such that $u(x_0, a_0) = 0$, meaning $(x_0, a_0)$ is an equilibrium point for the ODE.
  Note an equilibrium of open ODE with interface the terminal bundle $1 = 1$ is a closed ODE with a distinguished equilibrium point $x_0 \in X$.

  A map of equilibria of open ODEs is a pointed $C^1$ function $f:X \to X'$ which commutes with the dynamics, a fact that can be witnessed by the existence of a square as such in $\Lens\BunMan$
  In this way we get symmetric monoidal loose right module $\ODE_\bullet$ over the symmetric monoidal double category $\Lens\BunMan$.
\end{construction}

\begin{remark}
  Our specific choices of open manifolds, bundles, and regularity hypotheses on the maps between them are, from the point of view of the categorical treatment below, not essential.
  The very formal nature of our constructions makes it easy to swap this tangency with a different one, according to the needs of the user.
  The setup here loosely tracks Sontag's in \cite{sontagInputStateStability2008} and related stability theory literature.
\end{remark}

We now define certified Moore machines over this tangency.
In \cref{sec:gen.moore.machine} we obtained the fibration of 2-algebras $\Moore(\acat{P}\pi, \Lift) \fibto \Moore(\pi, T)$ by applying $\Moore$ directly to a fibration in tangencies; here we directly define the double category of certified lenses and its algebra of certified Moore machines.
The full construction is in \cref{app:lyapunov-dbl-cat-const} and amounts to a monoidal double Grothendieck construction.
We comment there on why we cannot directy apply the methods of \cref{sec:gen.moore.machine} in this case.
Abstract nonsense notwithstanding, the 2-algebra we get is still very much of the same \emph{flavor}, and we describe it now.

On manifolds, our generalized predicates will be \textbf{$C^{1,1}$ local storage functions} \cite{sontagInputStateStability2008,mironchenkoInputtoStateStabilityTheory2023} (from here on, also just \emph{predicates}), ordered by pointwise inequality.
More generally, a map $f : ((X', x'_0), \varphi) \to ((X, x_0), \psi)$ consists of a $C^{1,1}$ map $f : X \to X$  for which $\varphi(x') \geq \psi(f(x'))$ for all $x' \in X'$.
Note we can take arbitrary sums of local storage functions.

As for predicates over bundles of actions, we start by equipping $T\R = \pi_1 : \R \times \R \to \R$ with the \emph{lexicographic} order on pairs:
\begin{equation}\label{eq:lexico}
  (a, a') \succeq (b, b')
  \qquad\text{iff}\qquad
    a \geq b \quad \text{and}\quad (a = b \ \ \text{implies}\ \ a' \geq b').
\end{equation}
Intuitively, the second coordinate is an infinitesimal displacement and therefore can never overtake a finite displacement.

Now a \textbf{local storage bundle function} (from here on, just \emph{bundle predicate}) over a bundle $X \times F \to X$ is just a map of pointed bundles as below left, and an inequality between them is diagram:

These inherit the pointwise lexicographic order from $T\R$.

\begin{definition}%
Consider the following double category $\CertLens_\R$:
\begin{enumerate}
  \item The objects and tight maps are bundle predicates and their inequations, as in \eqref{eq:bund-preds}.
  \item The loose maps consist of lenses $\lens{w^\sharp}{w} : \lens{A_1}{O_1} \lensto \lens{A_2}{O_2}$ for the fibration of bundles over manifolds described in \cref{eq:man-bundles-fib} together with a $\K_{\infty}^0$ function $\kappa$ (which we call the \textbf{slack}) which together satisfy the following pair of conditions (the `certification' part):
  \begin{equation}
    \begin{cases}
      \alpha_2(w(o_1), a_2) + \kappa(\gamma_1(o_1)) \geq \alpha_1(w^\sharp(o_1, a_2))\\
      \gamma_1(o_1) \geq \gamma_2(w(o_1))
    \end{cases}
  \end{equation}
  We denote quantitatively certified lenses as we did in the previous sections, but annotated by $\kappa$:
  \begin{equation}
    \lens{w^\sharp}{w} \models \lens{\alpha_1}{\gamma_1} \overset{\kappa}{\lensto} \lens{\alpha_2}{\gamma_2}.
  \end{equation}
  These compose by lens composition and by addition of the slacks.
  \item There is a square of a given signature precisely when the underlying uncertified lenses and maps form a square of uncertified lenses, and the slack of the top lens dominates that of the bottom.
\end{enumerate}
\end{definition}

\begin{definition}\label{def:cert-ode}
  A \textbf{certified ODE} is a lens of the form $\lens{u}{v} : \lens{TX}{X} \lensto \lens{A}{O}$ certified by $\lens{\de \varphi + \varphi}{\varphi} \lensto \lens{\alpha}{\gamma}$.
\end{definition}

Note that if $\varphi \geq \psi :X \to \R$, then $(\varphi, \de \varphi + \varphi) \geq (\psi, \de \psi + \psi)$ \emph{lexicographically} but not with the product order.
This makes $\CertODE$ a well-defined 2-algebra over $\CertLens_\R$, by composition on the left, very reminiscent of a module of Moore machines.
It is moreover fibred over $\ODE_\bullet$.

\subsection{Lyapunov functions as certified open ODEs}
We have now the means to reformulate \cref{cor:main} as saying that \textbf{an equilibrium of open ODE is local ISS if and only if it is a certified ODE (in the sense of \cref{def:cert-ode})} of the form
\vspace{-.5ex}
\begin{equation}\label{eq:liss-ode}
  \lens{u}{v} \vDash \lens{\de \varphi + \varphi}{\varphi} \overset{\lambda}\lensto \lens{\alpha\pi_2}{\gamma}.
\end{equation}
\vspace{-1ex}
where
\begin{enumerate}
  \item $\alpha$ is a local storage function $A \to \R$, that is, it cannot depend on its first argument,
  \item the slack $\lambda$ is such that $\id-\lambda \in \K_\infty$ --- note that $\lambda=0$ has this property.
\end{enumerate}
Indeed, under these assumptions the certification amounts to
\begin{equation}
  \begin{cases}
    \alpha(a) \geq \de \varphi(f(x,a)) + (\id-\lambda)(\varphi(x))\\
    \varphi(x) \geq \gamma(v(x))
  \end{cases}
\end{equation}
We also see that then global ISS systems are those for which $\gamma$ is unbounded and $v$ has unbounded image.

\begin{remark}
  By traslation, we obtain similar characterizations even when the distinguished equilibrium $x_0 \neq 0$ --- we fixed $x_0=0$ only for convenience.
  Also, in light of \cref{rmk:k-funs} we can see that, up to change of variables \emph{on the interface}, we can always consider $\gamma$ and $\alpha$ to be the bare norm $|-|$.
\end{remark}

We can also revisit the two key properties of ISS, non-linearity and the dissipative characterization, and note they fit very naturally in our framework since they correspond to the soundness of the proof rules \ref{rule1} and \ref{rule2} in the 2-algebra of certified Moore machines, specifically as constructed in \cref{sec:cert-odes}.

Indeed, the fibrancy of $\CertODE_\bullet$ over $\ODE_\bullet$ corresponds to the stability of ISS under change of variables.
What this means is that if $X$ and $X'$ are equilibria of open ODEs and $X' \cong X$ is a morphism, then if $X'$ is ISS then so is $X$.
In fact, this works for any simulation $X' \to X$, not necessarily invertible ones.

As for compositionality, this straight up holds by construction, though for it to preserve ISS we must ensure the two aforementioned conditions are met, that is (1) the lens we use to compose the systems must have outer interface certified by an assumption of the form $\alpha\pi_2$ and (2) there is a non-trivial check on $\lambda$ --- we must have `enough slack' to pull through composition.

We illustrate it with an example:

\begin{example}
  In \cite[§4]{sontagInputStateStability2008}, it is shown that ISS systems compose sequentially. %
  Let $\lens{u_1}{v_1}$ and $\lens{u_2}{v_2}$ be certified ODEs with interface $\lens{A}{M} \vDash \lens{\alpha_1\pi_2}{\gamma_1}$ and $\lens{M}{O} \vDash \lens{\alpha_2\pi_2}{\gamma_2}$, respectively, and with slacks $\lambda_1$ and $\lambda_2$.
  The lens for sequential composition is $\lens{s^\sharp}{s} : \lens{A \times M}{M \times O} \lensto \lens{A}{O}$ defined as $s(m,o) = o$ and $s^\sharp(m,o,a) = (a,m)$ --- it is in fact a special case of feedback wiring.
  The canonical certification on $\lens{s^\sharp}{s}$ is as follows.
  First, we choose a slack $\kappa$ such that $\kappa(\gamma_1) \geq \alpha_1$ --- that is, we need the assumption on the first composee to imply the assumption on the second to which it feeds in.
  We have
  \begin{equation}
    \lens{s^\sharp}{s} \vDash \lens{(\alpha_1 + \alpha_2)\pi_2}{\gamma_1 + \gamma_2} \overset{\kappa}\lensto \lens{\alpha_1\pi_2}{\gamma_2}
  \end{equation}
  So we see that in order to proceed with composition we must have $\gamma_1(m)+\gamma_2(o) \geq \gamma_2(o)$ --- always true --- and $\alpha_2(m) + \kappa(\gamma_1(m) + \gamma_2(o)) \geq \alpha_2(m) + \kappa(\gamma_1(m)) \geq \alpha_1(a) + \alpha_2(m)$ --- true by assumption.
\end{example}

\section{Conclusion and Future Work}

We have described a general framework for assume-guarantee reasoning of state safety predicates for generalized Moore machines, as well as a special case adapted to LISS Lyapunov functions. We aim to expand on this work in a few ways.
First, we hope to further develop the various examples covered by our framework, most especially POMDPs.
Second, and more decisively, we aim to expand beyond state safety predicates to general (possibly quantitative) $\omega$-regular predicates of traces using supermartingale certificates \cite{abate2024stochasticomegaregularverificationcontrol, abate2025quantitativesupermartingalecertificates,henzinger2025supermartingalecertificatesquantitativeomegaregular}.
This more general verification works by coupling systems with a B\"uchi or Streett automaton that recognizes the predicate, and then giving a Lyapunov function on the coupled system; in fact, it is possible to express these ``coupled Lyapunov functions'' as morphisms on the original, uncoupled system \cite{Myers_djm00CP} (potentially leaving it open for assume-guarantee reasoning).
We believe the theory here can be extended thereby to give assume-guarantee reasoning for these general supermartingale certificates.

\bibliographystyle{eptcs}
\bibliography{main}

\begin{thebibliography}{10}
\providecommand{\bibitemdeclare}[2]{}
\providecommand{\surnamestart}{}
\providecommand{\surnameend}{}
\providecommand{\urlprefix}{Available at }
\providecommand{\url}[1]{\texttt{#1}}
\providecommand{\href}[2]{\texttt{#2}}
\providecommand{\urlalt}[2]{\href{#1}{#2}}
\providecommand{\doi}[1]{doi:\urlalt{https://doi.org/#1}{#1}}
\providecommand{\eprint}[1]{arXiv:\urlalt{https://arxiv.org/abs/#1}{#1}}
\providecommand{\bibinfo}[2]{#2}

\bibitemdeclare{misc}{abate2024stochasticomegaregularverificationcontrol}
\bibitem{abate2024stochasticomegaregularverificationcontrol}
\bibinfo{author}{Alessandro \surnamestart Abate\surnameend},
  \bibinfo{author}{Mirco \surnamestart Giacobbe\surnameend} \&
  \bibinfo{author}{Diptarko \surnamestart Roy\surnameend}
  (\bibinfo{year}{2024}): \emph{\bibinfo{title}{Stochastic Omega-Regular
  Verification and Control with Supermartingales}}.
\newblock \eprint{2405.17304}.

\bibitemdeclare{misc}{abate2025quantitativesupermartingalecertificates}
\bibitem{abate2025quantitativesupermartingalecertificates}
\bibinfo{author}{Alessandro \surnamestart Abate\surnameend},
  \bibinfo{author}{Mirco \surnamestart Giacobbe\surnameend} \&
  \bibinfo{author}{Diptarko \surnamestart Roy\surnameend}
  (\bibinfo{year}{2025}): \emph{\bibinfo{title}{Quantitative Supermartingale
  Certificates}}.
\newblock \eprint{2504.05065}.

\bibitemdeclare{misc}{ames2025categoricallyapunovtheoryii}
\bibitem{ames2025categoricallyapunovtheoryii}
\bibinfo{author}{Aaron~D. \surnamestart Ames\surnameend},
  \bibinfo{author}{Sébastien \surnamestart Mattenet\surnameend} \&
  \bibinfo{author}{Joe \surnamestart Moeller\surnameend}
  (\bibinfo{year}{2025}): \emph{\bibinfo{title}{Categorical Lyapunov Theory II:
  Stability of Systems}}.
\newblock \eprint{2505.22968}.

\bibitemdeclare{misc}{ames2025categoricallyapunovtheoryi}
\bibitem{ames2025categoricallyapunovtheoryi}
\bibinfo{author}{Aaron~D. \surnamestart Ames\surnameend}, \bibinfo{author}{Joe
  \surnamestart Moeller\surnameend} \& \bibinfo{author}{Paulo \surnamestart
  Tabuada\surnameend} (\bibinfo{year}{2025}): \emph{\bibinfo{title}{Categorical
  Lyapunov Theory I: Stability of Flows}}.
\newblock \eprint{2502.15276}.

\bibitemdeclare{misc}{arkor2024enhanced2categoricalstructurestwodimensional}
\bibitem{arkor2024enhanced2categoricalstructurestwodimensional}
\bibinfo{author}{Nathanael \surnamestart Arkor\surnameend},
  \bibinfo{author}{John \surnamestart Bourke\surnameend} \&
  \bibinfo{author}{Joanna \surnamestart Ko\surnameend} (\bibinfo{year}{2024}):
  \emph{\bibinfo{title}{Enhanced 2-categorical structures, two-dimensional
  limit sketches and the symmetry of internalisation}}.
\newblock \eprint{2412.07475}.

\bibitemdeclare{article}{bakirtzisCategoricalSemantics}
\bibitem{bakirtzisCategoricalSemantics}
\bibinfo{author}{Georgios \surnamestart Bakirtzis\surnameend},
  \bibinfo{author}{Cody~H. \surnamestart Fleming\surnameend} \&
  \bibinfo{author}{Christina \surnamestart Vasilakopoulou\surnameend}
  (\bibinfo{year}{2021}): \emph{\bibinfo{title}{Categorical Semantics of
  Cyber-Physical Systems Theory}}.
\newblock {\slshape \bibinfo{journal}{ACM Transactions on Cyber-Physical
  Systems,}} \bibinfo{volume}{5}(\bibinfo{number}{3}), \doi{10.1145/3461669}.
\newblock \urlprefix\url{https://doi.org/10.1145/3461669}.

\bibitemdeclare{inproceedings}{BobaruPasareanuGiannakopoulou2008}
\bibitem{BobaruPasareanuGiannakopoulou2008}
\bibinfo{author}{Mihaela~Gheorghiu \surnamestart Bobaru\surnameend},
  \bibinfo{author}{Corina~S. \surnamestart P{\u{a}}s{\u{a}}reanu\surnameend} \&
  \bibinfo{author}{Dimitra \surnamestart Giannakopoulou\surnameend}
  (\bibinfo{year}{2008}): \emph{\bibinfo{title}{Automated Assume-Guarantee
  Reasoning by Abstraction Refinement}}.
\newblock In: {\slshape \bibinfo{booktitle}{Proceedings of the 20th
  International Conference on Computer Aided Verification (CAV 2008)}},
  {\slshape \bibinfo{series}{Lecture Notes in Computer Science}}
  \bibinfo{volume}{5123}, \bibinfo{publisher}{Springer}, pp.
  \bibinfo{pages}{135--148}, \doi{10.1007/978-3-540-70545-1_14}.

\bibitemdeclare{inproceedings}{CobleighGiannakopoulouPasareanu2003}
\bibitem{CobleighGiannakopoulouPasareanu2003}
\bibinfo{author}{Jamieson~M. \surnamestart Cobleigh\surnameend},
  \bibinfo{author}{Dimitra \surnamestart Giannakopoulou\surnameend} \&
  \bibinfo{author}{Corina~S. \surnamestart P{\u {a}}s{\u{a}}reanu\surnameend}
  (\bibinfo{year}{2003}): \emph{\bibinfo{title}{Learning Assumptions for
  Compositional Verification}}.
\newblock In: {\slshape \bibinfo{booktitle}{Proceedings of the 9th
  International Conference on Tools and Algorithms for the Construction and
  Analysis of Systems (TACAS 2003)}}, {\slshape \bibinfo{series}{Lecture Notes
  in Computer Science}} \bibinfo{volume}{2619}, \bibinfo{publisher}{Springer},
  pp. \bibinfo{pages}{331--346}, \doi{10.1007/3-540-36577-X_24}.

\bibitemdeclare{article}{CobleighGiannakopoulouPasareanuBarringer2008}
\bibitem{CobleighGiannakopoulouPasareanuBarringer2008}
\bibinfo{author}{Jamieson~M. \surnamestart Cobleigh\surnameend},
  \bibinfo{author}{Dimitra \surnamestart Giannakopoulou\surnameend},
  \bibinfo{author}{Corina~S. \surnamestart P{\u {a}}s{\u{a}}reanu\surnameend}
  \& \bibinfo{author}{Howard \surnamestart Barringer\surnameend}
  (\bibinfo{year}{2008}): \emph{\bibinfo{title}{Learning to Divide and Conquer:
  Applying the {L}* Algorithm to Automate Assume-Guarantee Reasoning}}.
\newblock {\slshape \bibinfo{journal}{Formal Methods in System Design}}
  \bibinfo{volume}{32}(\bibinfo{number}{3}), pp. \bibinfo{pages}{175--205},
  \doi{10.1007/s10703-008-0049-6}.

\bibitemdeclare{article}{cruttwell2022doublefibrations}
\bibitem{cruttwell2022doublefibrations}
\bibinfo{author}{Geoffrey S.~H. \surnamestart Cruttwell\surnameend},
  \bibinfo{author}{Michael \surnamestart Lambert\surnameend},
  \bibinfo{author}{Dorette \surnamestart Pronk\surnameend} \&
  \bibinfo{author}{Martin \surnamestart Szyld\surnameend}
  (\bibinfo{year}{2022}): \emph{\bibinfo{title}{Double Fibrations}}.
\newblock {\slshape \bibinfo{journal}{Theory and Applications of Categories}}
  \bibinfo{volume}{38}(\bibinfo{number}{35}), pp. \bibinfo{pages}{1326--1394},
  \doi{10.48550/arXiv.2205.15240}.
\newblock \urlprefix\url{https://arxiv.org/abs/2205.15240}.

\bibitemdeclare{inproceedings}{ulrichPredicateLifting}
\bibitem{ulrichPredicateLifting}
\bibinfo{author}{Ulrich \surnamestart Dorsch\surnameend},
  \bibinfo{author}{Stefan \surnamestart Milius\surnameend},
  \bibinfo{author}{Lutz \surnamestart Schr{\"o}der\surnameend} \&
  \bibinfo{author}{Thorsten \surnamestart Wi{\ss}mann\surnameend}
  (\bibinfo{year}{2018}): \emph{\bibinfo{title}{Predicate Liftings and Functor
  Presentations in Coalgebraic Expression Languages}}.
\newblock In \bibinfo{editor}{Corina \surnamestart C{\^i}rstea\surnameend},
  editor: {\slshape \bibinfo{booktitle}{Coalgebraic Methods in Computer
  Science}}, \bibinfo{publisher}{Springer International Publishing},
  \bibinfo{address}{Cham}, pp. \bibinfo{pages}{56--77}.

\bibitemdeclare{inproceedings}{geatti2025automata}
\bibitem{geatti2025automata}
\bibinfo{author}{Luca \surnamestart Geatti\surnameend} (\bibinfo{year}{2025}):
  \emph{\bibinfo{title}{Automata Cascades for Model Checking}}.
\newblock In \bibinfo{editor}{Angelo \surnamestart Montanari\surnameend},
  \bibinfo{editor}{Andrea \surnamestart Orlandini\surnameend},
  \bibinfo{editor}{Nicola \surnamestart Saccomanno\surnameend} \&
  \bibinfo{editor}{Stefano \surnamestart Tonetta\surnameend}, editors:
  {\slshape \bibinfo{booktitle}{Short Paper Proceedings of the 7th
  International Workshop on Artificial Intelligence and Formal Verification,
  Logic, Automata , and Synthesis, {OVERLAY} 2025, Bologna, Italy, October 26,
  2025}}, {\slshape \bibinfo{series}{CEUR Workshop Proceedings}}
  \bibinfo{volume}{4142}, \bibinfo{publisher}{CEUR-WS.org}, pp.
  \bibinfo{pages}{49--57}.
\newblock \urlprefix\url{https://ceur-ws.org/Vol-4142/paper6.pdf}.

\bibitemdeclare{article}{gumbergModelChecking}
\bibitem{gumbergModelChecking}
\bibinfo{author}{Orna \surnamestart Grumberg\surnameend} \&
  \bibinfo{author}{David~E. \surnamestart Long\surnameend}
  (\bibinfo{year}{1994}): \emph{\bibinfo{title}{Model checking and modular
  verification}}.
\newblock {\slshape \bibinfo{journal}{ACM Trans. Program. Lang. Syst.}}
  \bibinfo{volume}{16}(\bibinfo{number}{3}), p. \bibinfo{pages}{843–871},
  \doi{10.1145/177492.177725}.
\newblock \urlprefix\url{https://doi.org/10.1145/177492.177725}.

\bibitemdeclare{article}{GuptaMcMillanFu2008}
\bibitem{GuptaMcMillanFu2008}
\bibinfo{author}{Anubhav \surnamestart Gupta\surnameend},
  \bibinfo{author}{Kenneth~L. \surnamestart McMillan\surnameend} \&
  \bibinfo{author}{Zhaohui \surnamestart Fu\surnameend} (\bibinfo{year}{2008}):
  \emph{\bibinfo{title}{Automated Assumption Generation for Compositional
  Verification}}.
\newblock {\slshape \bibinfo{journal}{Formal Methods in System Design}}
  \bibinfo{volume}{32}(\bibinfo{number}{3}), pp. \bibinfo{pages}{285--301},
  \doi{10.1007/s10703-008-0050-0}.

\bibitemdeclare{misc}{haugseng2023twovariablefibrationsfactorisationsystems}
\bibitem{haugseng2023twovariablefibrationsfactorisationsystems}
\bibinfo{author}{Rune \surnamestart Haugseng\surnameend},
  \bibinfo{author}{Fabian \surnamestart Hebestreit\surnameend},
  \bibinfo{author}{Sil \surnamestart Linskens\surnameend} \&
  \bibinfo{author}{Joost \surnamestart Nuiten\surnameend}
  (\bibinfo{year}{2023}): \emph{\bibinfo{title}{Two-variable fibrations,
  factorisation systems and $\infty $-categories of spans}}.
\newblock \eprint{2011.11042}.

\bibitemdeclare{misc}{henzinger2025supermartingalecertificatesquantitativeomegaregular}
\bibitem{henzinger2025supermartingalecertificatesquantitativeomegaregular}
\bibinfo{author}{Thomas~A. \surnamestart Henzinger\surnameend},
  \bibinfo{author}{Kaushik \surnamestart Mallik\surnameend},
  \bibinfo{author}{Pouya \surnamestart Sadeghi\surnameend} \&
  \bibinfo{author}{\surnamestart Đorđe Žikelić\surnameend}
  (\bibinfo{year}{2025}): \emph{\bibinfo{title}{Supermartingale Certificates
  for Quantitative Omega-regular Verification and Control}}.
\newblock \eprint{2505.18833}.

\bibitemdeclare{article}{HenzingerQadeerRajamaniTasiran2002}
\bibitem{HenzingerQadeerRajamaniTasiran2002}
\bibinfo{author}{Thomas~A. \surnamestart Henzinger\surnameend},
  \bibinfo{author}{Shaz \surnamestart Qadeer\surnameend},
  \bibinfo{author}{Sriram~K. \surnamestart Rajamani\surnameend} \&
  \bibinfo{author}{Serdar \surnamestart Ta{\c{s}}{\i}ran\surnameend}
  (\bibinfo{year}{2002}): \emph{\bibinfo{title}{An Assume-Guarantee Rule for
  Checking Simulation}}.
\newblock {\slshape \bibinfo{journal}{ACM Transactions on Programming Languages
  and Systems}} \bibinfo{volume}{24}(\bibinfo{number}{1}), pp.
  \bibinfo{pages}{51--64}, \doi{10.1145/509705.509707}.

\bibitemdeclare{article}{HERMIDA199983}
\bibitem{HERMIDA199983}
\bibinfo{author}{Claudio \surnamestart Hermida\surnameend}
  (\bibinfo{year}{1999}): \emph{\bibinfo{title}{Some properties of Fib as a
  fibred 2-category}}.
\newblock {\slshape \bibinfo{journal}{Journal of Pure and Applied Algebra}}
  \bibinfo{volume}{134}(\bibinfo{number}{1}), pp. \bibinfo{pages}{83--109},
  \doi{https://doi.org/10.1016/S0022-4049(97)00129-1}.
\newblock
  \urlprefix\url{https://www.sciencedirect.com/science/article/pii/S0022404997001291}.

\bibitemdeclare{book}{Jacobs_2016}
\bibitem{Jacobs_2016}
\bibinfo{author}{Bart \surnamestart Jacobs\surnameend} (\bibinfo{year}{2016}):
  \emph{\bibinfo{title}{Introduction to Coalgebra: Towards Mathematics of
  States and Observation}}.
\newblock \bibinfo{series}{Cambridge Tracts in Theoretical Computer Science},
  \bibinfo{publisher}{Cambridge University Press}.

\bibitemdeclare{article}{Jaz_Myers_2021}
\bibitem{Jaz_Myers_2021}
\bibinfo{author}{David \surnamestart Jaz~Myers\surnameend}
  (\bibinfo{year}{2021}): \emph{\bibinfo{title}{Double Categories of Open
  Dynamical Systems (Extended Abstract)}}.
\newblock {\slshape \bibinfo{journal}{Electronic Proceedings in Theoretical
  Computer Science}} \bibinfo{volume}{333}, p. \bibinfo{pages}{154–167},
  \doi{10.4204/eptcs.333.11}.
\newblock \urlprefix\url{http://dx.doi.org/10.4204/EPTCS.333.11}.

\bibitemdeclare{misc}{libkind2025doubleoperadictheorysystems}
\bibitem{libkind2025doubleoperadictheorysystems}
\bibinfo{author}{Sophie \surnamestart Libkind\surnameend} \&
  \bibinfo{author}{David~Jaz \surnamestart Myers\surnameend}
  (\bibinfo{year}{2025}): \emph{\bibinfo{title}{Towards a double operadic
  theory of systems}}.
\newblock \eprint{2505.18329}.

\bibitemdeclare{book}{mironchenkoInputtoStateStabilityTheory2023}
\bibitem{mironchenkoInputtoStateStabilityTheory2023}
\bibinfo{author}{Andrii \surnamestart Mironchenko\surnameend}
  (\bibinfo{year}{2023}): \emph{\bibinfo{title}{Input-to-{{State Stability}}:
  {{Theory}} and {{Applications}}}}.
\newblock \bibinfo{series}{Communications and {{Control Engineering}}},
  \bibinfo{publisher}{Springer International Publishing},
  \bibinfo{address}{Cham}, \doi{10.1007/978-3-031-14674-9}.
\newblock \urlprefix\url{https://link.springer.com/10.1007/978-3-031-14674-9}.

\bibitemdeclare{inbook}{mooreGedanken}
\bibitem{mooreGedanken}
\bibinfo{author}{Edward~F. \surnamestart Moore\surnameend}
  (\bibinfo{year}{2016}): \emph{\bibinfo{title}{Gedanken-Experiments on
  Sequential Machines}}, pp. \bibinfo{pages}{129--154}.
\newblock \bibinfo{publisher}{Princeton University Press},
  \bibinfo{address}{Princeton}, \doi{doi:10.1515/9781400882618-006}.
\newblock \urlprefix\url{https://doi.org/10.1515/9781400882618-006}.

\bibitemdeclare{misc}{Myers_djm00CP}
\bibitem{Myers_djm00CP}
\bibinfo{author}{David~Jaz \surnamestart Myers\surnameend}:
  \emph{\bibinfo{title}{On the representability of Lyapunov-type functions}}.
\newblock \urlprefix\url{https://forest.topos.site/public/djm-00CP/}.
\newblock \bibinfo{note}{Unpublished note}.

\bibitemdeclare{misc}{myers2021cartesianfactorizationsystemsgrothendieck}
\bibitem{myers2021cartesianfactorizationsystemsgrothendieck}
\bibinfo{author}{David~Jaz \surnamestart Myers\surnameend}
  (\bibinfo{year}{2021}): \emph{\bibinfo{title}{Cartesian Factorization Systems
  and Grothendieck Fibrations}}.
\newblock \eprint{2006.14022}.

\bibitemdeclare{misc}{jaz-2021-book}
\bibitem{jaz-2021-book}
\bibinfo{author}{David~Jaz \surnamestart Myers\surnameend}
  (\bibinfo{year}{2021}): \emph{\bibinfo{title}{Categorical Systems Theory}}.
\newblock \urlprefix\url{http://davidjaz.com/Papers/DynamicalBook.pdf}.

\bibitemdeclare{inproceedings}{PasareanuDwyerHuth1999}
\bibitem{PasareanuDwyerHuth1999}
\bibinfo{author}{Corina~S. \surnamestart P{\u{a}}s{\u{a}}reanu\surnameend},
  \bibinfo{author}{Matthew~B. \surnamestart Dwyer\surnameend} \&
  \bibinfo{author}{Michael \surnamestart Huth\surnameend}
  (\bibinfo{year}{1999}): \emph{\bibinfo{title}{Assume-Guarantee Model Checking
  of Software: A Comparative Case Study}}.
\newblock In: {\slshape \bibinfo{booktitle}{Proceedings of the 6th
  International SPIN Workshop on Theoretical and Practical Aspects of SPIN
  Model Checking}}, {\slshape \bibinfo{series}{Lecture Notes in Computer
  Science}} \bibinfo{volume}{1680}, \bibinfo{publisher}{Springer}, pp.
  \bibinfo{pages}{168--183}, \doi{10.1007/3-540-48234-2_14}.

\bibitemdeclare{inproceedings}{pneuliAssumeGuarantee}
\bibitem{pneuliAssumeGuarantee}
\bibinfo{author}{Amir \surnamestart Pnueli\surnameend} (\bibinfo{year}{1985}):
  \emph{\bibinfo{title}{In Transition From Global to Modular Temporal Reasoning
  about Programs}}.
\newblock In \bibinfo{editor}{Krzysztof~R. \surnamestart Apt\surnameend},
  editor: {\slshape \bibinfo{booktitle}{Logics and Models of Concurrent
  Systems}}, \bibinfo{publisher}{Springer Berlin Heidelberg},
  \bibinfo{address}{Berlin, Heidelberg}, pp. \bibinfo{pages}{123--144}.

\bibitemdeclare{misc}{rupel2013operadtemporalwiringdiagrams}
\bibitem{rupel2013operadtemporalwiringdiagrams}
\bibinfo{author}{Dylan \surnamestart Rupel\surnameend} \&
  \bibinfo{author}{David~I. \surnamestart Spivak\surnameend}
  (\bibinfo{year}{2013}): \emph{\bibinfo{title}{The operad of temporal wiring
  diagrams: formalizing a graphical language for discrete-time processes}}.
\newblock \eprint{1307.6894}.

\bibitemdeclare{article}{RUTTEN20003}
\bibitem{RUTTEN20003}
\bibinfo{author}{J.J.M.M. \surnamestart Rutten\surnameend}
  (\bibinfo{year}{2000}): \emph{\bibinfo{title}{Universal coalgebra: a theory
  of systems}}.
\newblock {\slshape \bibinfo{journal}{Theoretical Computer Science}}
  \bibinfo{volume}{249}(\bibinfo{number}{1}), pp. \bibinfo{pages}{3--80},
  \doi{https://doi.org/10.1016/S0304-3975(00)00056-6}.
\newblock
  \urlprefix\url{https://www.sciencedirect.com/science/article/pii/S0304397500000566}.
\newblock \bibinfo{note}{Modern Algebra}.

\bibitemdeclare{misc}{schultz2019dynamicalsystemssheaves}
\bibitem{schultz2019dynamicalsystemssheaves}
\bibinfo{author}{Patrick \surnamestart Schultz\surnameend},
  \bibinfo{author}{David~I. \surnamestart Spivak\surnameend} \&
  \bibinfo{author}{Christina \surnamestart Vasilakopoulou\surnameend}
  (\bibinfo{year}{2019}): \emph{\bibinfo{title}{Dynamical Systems and
  Sheaves}}.
\newblock \eprint{1609.08086}.

\bibitemdeclare{article}{sontagSmoothStabilizationImplies1989}
\bibitem{sontagSmoothStabilizationImplies1989}
\bibinfo{author}{Eduardo~D. \surnamestart Sontag\surnameend}
  (\bibinfo{year}{1989}): \emph{\bibinfo{title}{Smooth Stabilization Implies
  Coprime Factorization}}.
\newblock {\slshape \bibinfo{journal}{IEEE transactions on automatic control}}
  \bibinfo{volume}{34}(\bibinfo{number}{4}), pp. \bibinfo{pages}{435--443}.
\newblock
  \urlprefix\url{https://www.academia.edu/download/48193612/Smooth_Stabilization_Implies_Coprime_Fac20160820-11625-15xs30o.pdf}.

\bibitemdeclare{inbook}{sontagInputStateStability2008}
\bibitem{sontagInputStateStability2008}
\bibinfo{author}{Eduardo~D. \surnamestart Sontag\surnameend}
  (\bibinfo{year}{2008}): \emph{\bibinfo{title}{Input to {{State Stability}}:
  {{Basic Concepts}} and {{Results}}}}, pp. \bibinfo{pages}{163--220}.
\newblock \bibinfo{volume}{1932}, \bibinfo{publisher}{Springer Berlin
  Heidelberg}, \bibinfo{address}{Berlin, Heidelberg},
  \doi{10.1007/978-3-540-77653-6_3}.
\newblock \urlprefix\url{http://link.springer.com/10.1007/978-3-540-77653-6_3}.

\bibitemdeclare{misc}{spivak2022generalizedlenscategoriesfunctors}
\bibitem{spivak2022generalizedlenscategoriesfunctors}
\bibinfo{author}{David~I. \surnamestart Spivak\surnameend}
  (\bibinfo{year}{2022}): \emph{\bibinfo{title}{Generalized Lens Categories via
  functors $\mathcal{C}^{\rm op}\to \mathsf{Cat}$}}.
\newblock \eprint{1908.02202}.

\bibitemdeclare{misc}{vagner2015algebrasopendynamicalsystems}
\bibitem{vagner2015algebrasopendynamicalsystems}
\bibinfo{author}{Dmitry \surnamestart Vagner\surnameend},
  \bibinfo{author}{David~I. \surnamestart Spivak\surnameend} \&
  \bibinfo{author}{Eugene \surnamestart Lerman\surnameend}
  (\bibinfo{year}{2015}): \emph{\bibinfo{title}{Algebras of Open Dynamical
  Systems on the Operad of Wiring Diagrams}}.
\newblock \eprint{1408.1598}.

\bibitemdeclare{inproceedings}{ViswanathanViswanathan2001}
\bibitem{ViswanathanViswanathan2001}
\bibinfo{author}{Mahesh \surnamestart Viswanathan\surnameend} \&
  \bibinfo{author}{Ramesh \surnamestart Viswanathan\surnameend}
  (\bibinfo{year}{2001}): \emph{\bibinfo{title}{Foundations for Circular
  Compositional Reasoning}}.
\newblock In: {\slshape \bibinfo{booktitle}{Proceedings of the 28th
  International Colloquium on Automata, Languages and Programming (ICALP
  2001)}}, {\slshape \bibinfo{series}{Lecture Notes in Computer Science}}
  \bibinfo{volume}{2076}, \bibinfo{publisher}{Springer}, pp.
  \bibinfo{pages}{142--153}, \doi{10.1007/3-540-48224-5_68}.

\end{thebibliography}

\appendix

\section{Construction of the double category $\CertLens_\R$}\label[appendix]{app:lyapunov-dbl-cat-const}

We construct the lax double functor $\Cert : \Lens\BunMan^{\mathsf{op}} \to \Span(\Cat)$ whose double Grothendieck construction (as in Section 3 of \cite{cruttwell2022doublefibrations}) is $\CertLens_{\R}$.

\begin{definition}
	We define a lax symmetric monoidal lax double functor
  \begin{equation}\Cert : \Lens\BunMan^{\mathsf{op}} \to \Span(\Cat)\end{equation}
  as follows:
	\begin{enumerate}
		\item $\Cert\lens{F}{B} := \BunMan\left(\lens{F}{B},\, \lens{\R}{\R} \right)$ is the set of bundle maps from $\pi_1 : B \times F \to B$ to $\pi_1 : \R \times \R \to \R$, ordered lexicographically:
    \begin{equation}\label{eq:bund-preds}
      \begin{tikzcd}[ampersand replacement=\&,column sep=scriptsize, row sep = tiny]
        \&\& {X' \times F'} \& \\
        {X \times F} \&\&\& {T\R} \\
        \&\& {X'} \\
        X \&\&\& \R
        \arrow[""{name=0, anchor=center, inner sep=0}, "{\alpha'}"{description}, curve={height=-6pt}, from=1-3, to=2-4]
        \arrow[two heads, from=1-3, to=3-3]
        \arrow["{f_\sharp}"{description}, from=2-1, to=1-3]
        \arrow[""{name=1, anchor=center, inner sep=0}, "\alpha"{description}, curve={height=6pt}, from=2-1, to=2-4]
        \arrow[two heads, from=2-1, to=4-1]
        \arrow[two heads, from=2-4, to=4-4]
        \arrow[""{name=2, anchor=center, inner sep=0}, "{\gamma'}"{description}, curve={height=-6pt}, from=3-3, to=4-4]
        \arrow["f"{description}, from=4-1, to=3-3]
        \arrow[""{name=3, anchor=center, inner sep=0}, "\gamma"{description}, curve={height=6pt}, from=4-1, to=4-4]
        \arrow["{\geq_\lexicographic}"{marking, allow upside down, pos=0.4}, shift left=3, draw=none, from=1, to=0]
        \arrow["\geq"{marking, allow upside down, pos=0.4}, shift left=3, draw=none, from=3, to=2]
      \end{tikzcd}
    \end{equation}
    These maps are given by pairs $\gamma : B\to \R$ and $\alpha : B \times F \to \R$, and the ordering has $(\gamma, \alpha) \to (\gamma', \alpha')$ when for all $x$, $\gamma(x) \geq \gamma'(x)$ and for all $x$ and $y$, if $\gamma(x) = \gamma'(x)$ then $\alpha(x, y) \geq \alpha'(x, y)$.

		\item For a bundle map $\lens{f_{\sharp}}{f} : \lens{F}{B} \rightrightarrows \lens{F'}{B'}$, $\Cert$ acts by precomposition. Explicitly, this sends $\lens{\alpha}{\gamma}$ to $\lens{(x, y) \mapsto \alpha(f(x), f_{\sharp}(x, y))}{x \mapsto \gamma(f(x))}$. This is evidently functorial.

		\item For a lens $\lens{w^{\sharp}}{w} : \lens{F}{B} \leftrightarrows \lens{F'}{G'}$, we define $\Cert\lens{w^{\sharp}}{w} \subseteq \K^0_{\infty} \times \Cert\lens{F}{B} \times \Cert\lens{F'}{B'}$ to be the full subcategory consisting of triples $\left(\kappa,\, \lens{\alpha}{\gamma},\, \lens{\alpha'}{\gamma'}\right)$ for which the lens is certified:
		\begin{equation}\label{eqn:certified.quantitative.lens}
			\lens{w^{\sharp}}{w} \vDash \lens{\alpha}{\gamma} \overset{\kappa}{\lensto} \lens{\alpha'}{\gamma'} :\iff \begin{cases}
				\alpha'(w(x), y) + \kappa(\gamma(x)) \geq \alpha(x, w^{\sharp}(x, y)) \\
				\gamma(x) \geq \gamma'(w(x))
			\end{cases}
		\end{equation}
		We refer to $\kappa$ as the \emph{slack}. Slacks are considered a category by $\kappa_1 \to \kappa_2$ when $\kappa_1 \geq \kappa_2$.
		\item For a square as below left, implying the equations below right:
		\begin{equation}
\begin{tikzcd}
	{\lens{F_1}{B_1}} & {\lens{F_3}{B_3}} \\
	{\lens{F_2}{B_2}} & {\lens{F_4}{B_4}}
	\arrow[shift right, from=1-1, to=1-2]
	\arrow[shift left, from=1-1, to=2-1]
	\arrow["{\lens{f_{\sharp}}{f}}"', shift right, from=1-1, to=2-1]
	\arrow["{\lens{u^{\sharp}}{u}}"', shift right, from=1-2, to=1-1]
	\arrow["{\lens{g_{\sharp}}{g}}", shift left, from=1-2, to=2-2]
	\arrow[shift right, from=1-2, to=2-2]
	\arrow["{\lens{w^{\sharp}}{w}}"', shift right, from=2-1, to=2-2]
	\arrow[shift right, from=2-2, to=2-1]
\end{tikzcd}\quad\quad
\begin{cases}
	w^{\sharp}(f(x), g_{\sharp}(u(x), y)) = f_{\sharp}(x, u^{\sharp}(x, y)) \\
	w(f(x)) = g(u(x))
\end{cases}
		\end{equation}
		we need to check that if $\lens{w^{\sharp}}{w} \vDash \lens{\alpha}{\gamma} \overset{\kappa}{\lensto} \lens{\alpha'}{\gamma'}$, then $\lens{u^{\sharp}}{u} \vDash \lens{f_{\sharp}}{f}^{\ast}\lens{\alpha}{\gamma} \overset{\kappa}{\lensto} \lens{g_{\sharp}}{g}^{\ast}\lens{\alpha'}{\gamma'}$; in other words, that
		\begin{equation}
			\begin{cases}
				 \alpha'(g(u(x)), g_{\sharp}(u(x), y)) + \kappa(\gamma(f(x))) \geq \alpha(f(x), f_{\sharp}(x, u^{\sharp}(f(x), y)))\\
				\gamma(f(x)) \geq \gamma'(g(u(x)))
			\end{cases}
		\end{equation}
		These follow formally from the equations governing the square and the assumption that $\lens{w^{\sharp}}{w} \vDash \lens{\alpha}{\gamma} \leftrightarrows \lens{\alpha'}{\gamma'}$, substituting in $x \leftarrow f(x)$ and $y \leftarrow g_{\sharp}(u(x), y)$ to \eqref{eqn:certified.quantitative.lens}.
		\item We then need to supply the unitor and compositor. We quite trivially have $\lens{\pi_2}{\id} \vDash \lens{\alpha}{\gamma} \overset{0}{\leftrightarrows} \lens{\alpha}{\gamma}$, which handles the unitor. For the compositor, we add the slacks; suppose that $\lens{w^{\sharp}_1}{w_1} \vDash \lens{\alpha_1}{\gamma_1} \overset{\kappa_1}{\leftrightarrows} \lens{\alpha_2}{\gamma_2}$ and $\lens{w^{\sharp}_2}{w_2} \vDash \lens{\alpha_2}{\gamma_2} \overset{\kappa_2}{\leftrightarrows} \lens{\alpha_3}{\gamma_3}$, seeking to show that $\lens{w^{\sharp}_2}{w_2} \circ \lens{w^{\sharp}_1}{w_1} \vDash \lens{\alpha_1}{\gamma_1} \overset{\kappa_1 + \kappa_2}{\leftrightarrows} \lens{\alpha_2}{\gamma_2}$. We may prove this as follows:
		\begin{equation}
\begin{cases}
	\begin{aligned}
	\alpha_3(w_2(w_1(x)), y) +  (\kappa_1 + \kappa_2)(\gamma_1(x)) &\geq \alpha_3(w_2(w_1(x)), y) + \kappa_1(\gamma_1(x)) + \kappa_2(\gamma_2(w_1(x)))  \\
	&\geq \alpha_2(w_1(x), w_2^{\sharp}(w_1(x), y)) + \kappa_2(\gamma_2(w_1(x))) \\
	 &\geq \alpha_1(x, w_1^{\sharp}(x, w^{\sharp}_2(w(x), y)))\\
	\end{aligned}\\
	\gamma_1(x) \geq \gamma_2(w_1(x)) \geq \gamma_3(w_2(w_1(x)))
\end{cases}
		\end{equation}
		\item The laxitor $\Cert \lens{F}{B} \times \Cert \lens{F'}{B'} \to \Cert \lens{F \times F'}{B \times B'}$ is given by componentwise addition:
		$$\left(\lens{\alpha}{\gamma}, \lens{\alpha'}{\gamma'}\right) \mapsto \lens{\alpha + \alpha'}{\gamma + \gamma'}.$$
		On lenses, we also add the slacks. This is evidently commutative, associative, and monotone. The identity is $\lens{0}{0} : \lens{1}{1} \rightrightarrows \lens{\R}{\R}$. It interchanges with compositors because both add the slacks and addition is commutative.
		\end{enumerate}
\end{definition}

We see thus that this double category is very close to be a double category of lenses, except that:
\begin{enumerate}
  \item The lexicographic order prevents bundle predicates from being fibred over predicates `in the correct way'. The map of posets $(\R \times \R, \geq_\lexicographic) \to (\R, \geq)$ is a Grothendieck fibration but pullback along a strict inequality $a \gneq b$ sends every pair $(b,b')$ to $(a,0)$.
  This is not the behaviour we want, what we want is to get $(a,b')$ or at most $(a, \eta(b'))$ where $\eta$ is parallel transport for a chosen connection (remember $b'$ is a tangent vector at $b$, so one must explicitly move it to a different fiber).
  \item One could fix the first problem by using the product order instead of lexicographic, but this in turns makes \cref{def:cert-ode} unsatisfactory, since $\varphi \mapsto (\varphi,\de \varphi)$ is not a monotone section of $(\R \times \R,\ \geq \times \geq) \to (\R, \geq)$.
  \item Finally, \emph{slack} is needed, in practice, for good compositional properties, and that cannot be accounted for in any way in a lens construction.
\end{enumerate}

\section{Proofs of \cref{sec:gen.moore.machine}}
\label[appendix]{app:fibs-of-tans}
We will make use of the theory of \emph{enhanced sketches} \cite{arkor2024enhanced2categoricalstructurestwodimensional}. Example 5.15 of \emph{ibid.} gives the enhanced sketch $\mathbb{T}_{\mathsf{fib}}$ for tight fibrations, so that models in the (chordate) enhanced 2-category of categories $\mathcal{M}\mathsf{od}_{s, l}(\mathbb{T}_{\mathsf{fib}}, \Cat) \cong \Fib$ are fibrations (lax- and pseudo-morphisms of fibrations coincide --- both are cartesian functors --- see Example 4.10 of \emph{ibid.}). We may define the sketch $\mathbb{T}_{\mathsf{tan}}$ of tangencies to contain the sketch for tight fibrations together with a \emph{loose} section; we then have $\Tangency \cong \mathcal{M}\mathsf{od}_{s, c}(\mathbb{T}_{\mathsf{Tan}}, \Cat)$, with maps given by the (strictly commuting) cartesian squares which are \emph{colax} on sections. We may therefore use \emph{symmetry of internalization} (Theorem 7.5 of \cite{arkor2024enhanced2categoricalstructurestwodimensional}) to prove \cref{lem:fibration.of.tangencies}.

\begin{proof}[Proof of \cref{lem:fibration.of.tangencies}]
By \emph{symmetry of internalization}, we have
\begin{equation}
\Fib(\Tangency) =
\Mod{l}(\mathbb{T}_{\mathsf{fib}}, \Mod{c}(\mathbb{T}_{\mathsf{tan}}, \Cat))
\cong
\Mod{c}(\mathbb{T}_{\mathsf{tan}}, \Mod{l}(\mathbb{T}_{\mathsf{fib}}, \Cat)) = \Tangency(\Fib)
\end{equation}
This shows that a tight fibration internal to the enhanced 2-category of tangencies is equivalently a tangency internal to the enhanced 2-category of fibrations. The tight morphisms of tangencies are the strict ones; the tight morphisms of fibrations strictly preserve the cleavage. A tangency internal to fibrations is a diagram of the form:
	\begin{equation}
\begin{tikzcd}[sep=small]
	{\acat{PB}} & {\acat{PE}} & {\acat{PB}} \\
	{\acat{B}} & {\acat{E}} & {\acat{B}}
	\arrow["\Lift", from=1-1, to=1-2]
	\arrow["{p_{\acat{B}}}"', two heads, from=1-1, to=2-1]
	\arrow["{\acat{P}\pi}", two heads, from=1-2, to=1-3]
	\arrow["{p_{\acat{E}}}"', two heads, from=1-2, to=2-2]
	\arrow["{p_{\acat{B}}}"', two heads, from=1-3, to=2-3]
	\arrow["T"', from=2-1, to=2-2]
	\arrow["\pi"', two heads, from=2-2, to=2-3]
\end{tikzcd}
	\end{equation}
	where the left square is cartesian (loose) and the right strictly preserves chosen cartesian lifts (tight); it remains to show that asking for $(\pi, \acat{P}\pi)$ to be a fibration internal to $\Fib$ is merely asking that it be componentwise a fibration. This is a classical fact; it is proved for morphisms of fibrations over a fixed base in Theorem 4.16 of \cite{HERMIDA199983} (attributed to B{\'e}nabou therein), and we may reduce our case to this by noting that a square is cartesian when the gap map into the pullback over the map on bases is.
\end{proof}

We check that sets and boolean valued predicates gives an example of this notion.

\begin{proof}[Proof of \cref{lem:example.fibration.tangency}]
The first thing to check is that $\mathsf{cod} : \thecat{PSet}^{\downarrow} \to \thecat{PSet}$ is a fibration; equivalently, this means that $\thecat{PSet}$ has pullbacks. Since $\Set$ has pullbacks and $\Set(X, \mathsf{bool})$ has pullbacks (given by conjunction $\wedge$) for all $X$, and since these are preserved under reindexing (precomposition), $\thecat{PSet}$ has pullbacks constructed by
\begin{equation}
\begin{tikzcd}[sep=small]
	{(A \times_C B, \pi_1^{\ast}\alpha \wedge \pi_2^{\ast} \beta)} & {(B, \beta)} \\
	{(A, \alpha)} & {(C, \gamma)}
	\arrow[from=1-1, to=1-2]
	\arrow["{\pi_1}"', from=1-1, to=2-1]
	\arrow["\lrcorner"{anchor=center, pos=0.125}, draw=none, from=1-1, to=2-2]
	\arrow["g", from=1-2, to=2-2]
	\arrow["f"', from=2-1, to=2-2]
\end{tikzcd}
\end{equation}
That $p^{\downarrow}$ is a fibration follows by Proposition 4.4 of \cite{HERMIDA199983}. We then need to check that $(\mathsf{cod}, \mathsf{cod}) : p^{\downarrow} \to p$ strictly preserves chosen lifts; but chosen lifts in each case are given by precomposition, so this is easily satisfied. Finally, we need to show that $\Lift (S, \varphi) := ((\pi_1 : (S \times S, \varphi \times \varphi) \to (S, \varphi)))$ is cartesian; but if $f : S' \to S$, then the lift is $f : (S', \varphi \circ f) \to (S, \varphi)$ and we have $(\varphi \circ f) \times (\varphi \circ f) = (\varphi \times \varphi) \circ (f \times f)$.

\end{proof}

Finally, we come to \cref{thm:moore.preserves.fibrations}. The key observation is that $\Lens : \Fib \to \Dbl$ preserves tight fibrations because it preserves slicing; from this the extension to $\Moore$ is straightforward.

\begin{proof}[Proof of \cref{thm:moore.preserves.fibrations}]
We first show that $\Lens : \Fib \to \Dbl$ preserves tight fibrations.
Since tight fibrations are given by a sketch only marking colax limits of tight arrows, it will suffice to show that $\Lens$ preserves colax limits of tight arrows (also known as \emph{slices} $X \downarrow f$). The 2-functor $\Lens$ is a composite
\begin{equation}
\Fib \xto{(\mathsf{dom}(-), (\mathsf{vert}, \mathsf{cart}))} \AdTriple \xto{\mathbb{S}\mathbf{pan}} \Dbl
\end{equation}
given first by sending a fibration $\pi : \acat{E} \to \acat{B}$ to the \emph{adequate triple} \cite{haugseng2023twovariablefibrationsfactorisationsystems} $(\acat{E}, (\mathsf{vert}, \mathsf{cart}))$, and then taking the span construction of adequate triples (see Section 2.3 of \cite{libkind2025doubleoperadictheorysystems} for a review). In Theorem 2.13 of \cite{haugseng2023twovariablefibrationsfactorisationsystems}, $\mathbb{S}\mathbf{pan} : \AdTriple \to \Dbl$ is shown to be a \emph{nerve} against the cosimplicial adequate triple $\acat{Tw}^r : \Delta \to \AdTriple$ sending each $n$-simplex to its twisted arrow category (see Example 2.9 of \emph{ibid.}); here, we take double categories $\Dbl \hookrightarrow \Cat^{\Delta^{\mathsf{op}}}$ to be simplicial categories satisfying the Segal condition. For this reason, $\mathbb{S}\mathbf{pan}$ preserves all 2-limits.

It therefore suffices to show that the functor sending $\pi : \acat{E} \to \acat{B}$ to $(\acat{E}, (\mathsf{vert}, \mathsf{cart}))$ preserves colax limits of tight arrows. By Proposition 4.14 of \cite{HERMIDA199983}, $\Fib$ has colax limits of arrows (in fact, all comma objects), and they are constructed componentwise in $\Cat$; moreover, a map in the slice is vertical when it is componentwise vertical, and cartesian when it is componentwise cartesian.
It therefore suffices to show that $\AdTriple$ also has colax limits of arrows constructed in $\Cat$. We address this in \cref{lem:slices.of.adequate.triples}.

Finally, we must show that $\Moore$ also preserves fibrations. Fibrations of loose right modules are constructed componentwise in $\Cat$, just like for $\Dbl$ (both double categories and loose right modules of double categories are sketchable). Suppose we have a fibration $(p_{\acat{E}}, p_{\acat{B}}) : (\acat{P}\pi, \Lift) \to (\pi, T)$ of tangencies. We have seen that $(p_{\acat{E}}, p_{\acat{B}})$ induces a fibration $p_{\acat{E},\star} : \Lens(\acat{P}\pi) \to \Lens(\pi)$; it remains to show that this extends to $\Moore(\acat{P}\pi, \Lift) \to \Moore(\pi, T)$, and for this it suffices to show on the categories of Moore machines. But a map of $(\pi, T)$-Moore machines is just a square of lenses whose tight domain is of the form $T\sigma$; it therefore suffices to note that a lift of such a square has tight domain $\Lift (\hat{\sigma})$, where $\hat{\sigma}$ is the lift of $\sigma$ along $p_{\acat{B}}$. This follows from the assumption that $\Lift$ preserves cartesian lifts.
 \end{proof}

\begin{lemma}\label{lem:slices.of.adequate.triples}
The 2-category $\AdTriple$ of adequte triples (see Definition 2.9 of \cite{libkind2025doubleoperadictheorysystems}) has all colax limits of arrows, and they are constructed in $\Cat$.
\end{lemma}
\begin{proof}
Let $\pi : (\acat{E}, (L_{\acat{E}}, R_{\acat{E}})) \to (\acat{B}, (L_{\acat{B}}, R_{\acat{B}}))$ be a morphism of adequate triples. We will show that the colax limit (i.e. slice or comma) $\acat{B} \downarrow \pi$ of the functor underlying $\pi$ may be endowed with the structure of an adequate triple, making it the colax limit of $\pi$ in $\AdTriple$. This is straightforward by taking the structure componentwise.

A morphism in $\acat{B} \downarrow \pi$ is a pair of morphisms in $\acat{B}$ and $\acat{E}$ such that a square commutes in $\acat{B}$; we'll take the left class to consist of those pairs where both maps are in their respective left classes, and similarly for the right class. This endows the slice $\acat{B} \downarrow \pi$ with the structure of an adequate triple by taking pullbacks componentwise, since $\pi$ preserves $L$-$R$ pullbacks; the pullback of $(\ell_0, \ell_1)$ along $(r_0, r_2)$ is given by:
\begin{equation}%
\begin{tikzcd}[sep=small]
	{B' \times_B B''} && {\pi(E' \times_E E'')} & \\
	& {B''} && {\pi E''} \\
	{B'} && {\pi E'} \\
	& B && {\pi E}
	\arrow["{x' \times_x x''}", dashed, from=1-1, to=1-3]
	\arrow[from=1-1, to=2-2]
	\arrow[from=1-1, to=3-1]
	\arrow["\lrcorner"{anchor=center, pos=0.125}, draw=none, from=1-1, to=4-2]
	\arrow[from=1-3, to=2-4]
	\arrow[from=1-3, to=3-3]
	\arrow["\lrcorner"{anchor=center, pos=0.125}, draw=none, from=1-3, to=4-4]
	\arrow["{x''}"{pos=0.3}, from=2-2, to=2-4]
	\arrow["{r_0}"{pos=0.7}, from=2-2, to=4-2]
	\arrow["{\pi r_1}"{pos=0.7}, from=2-4, to=4-4]
	\arrow["{x'}"{pos=0.7}, from=3-1, to=3-3]
	\arrow["{\ell_0}"', from=3-1, to=4-2]
	\arrow["{\pi \ell_1}"{pos=0.4}, from=3-3, to=4-4]
	\arrow["x"', from=4-2, to=4-4]
\end{tikzcd}
 \end{equation}
To show that this is the colax limit of $\pi$ in $\AdTriple$, we must show that if $\alpha : f \Rightarrow \pi g$ is a natural transformation and $f : (\acat{X}, (L_{\acat{X}}, R_{\acat{X}})) \to
(\acat{B}, (L_{\acat{B}}, R_{\acat{B}}))$ and $g :(\acat{X}, (L_{\acat{X}}, R_{\acat{X}})) \to
(\acat{E}, (L_{\acat{E}}, R_{\acat{E}}))$ are morphisms of adequate triples, then the induced functor $\langle \alpha \rangle : \acat{X} \to \acat{B} \downarrow \pi$ is also a morphism of adequate triples; but by assumption $f$ and $g$ preserve both classes and $L$-$R$ pullbacks, and the adequate triple structure of $\acat{B} \downarrow \pi$ was chosen to be componentwise.
\end{proof}

\end{document}